\documentclass[journal, draftcls, onecolumn,12pt,twoside]{IEEEtranTCOM}

\usepackage{times}
\usepackage{latexsym}
\usepackage{stfloats}
\usepackage{citesort}
\usepackage{graphicx, amsfonts,amsmath, amssymb,array,url,stfloats}
\usepackage[noadjust]{cite}
\usepackage{graphicx}
\usepackage{subfigure}
\usepackage{enumerate}
\usepackage{hyperref}
\usepackage{setspace}
\usepackage{color}
\doublespace

\newtheorem{theorem}{Theorem}
\newtheorem{corollary}{Corollary}
\newtheorem{proposition}{Proposition}
\newtheorem{lemma}{Lemma}

\newtheorem{remark}{Remark}

\makeatletter

\begin{document}

\title{Impact of Residual Transmit RF Impairments on Training-Based MIMO Systems}

\author{Xinlin~Zhang,~\IEEEmembership{Student Member,~IEEE,}
        Michail~Matthaiou,~\IEEEmembership{Senior Member,~IEEE,}
        Mikael~Coldrey,~\IEEEmembership{Member,~IEEE,}
        and~Emil~Bj\"ornson,~\IEEEmembership{Member,~IEEE}
\thanks{X. Zhang is with the Department
of Signals and Systems, Chalmers University of Technology, 412 96,
Gothenburg, Sweden (email: xinlin@chalmers.se).}
\thanks{M. Matthaiou is with the School of Electronics, Electrical Engineering and Computer Science, Queen's University Belfast, Belfast, U.K. and with the Department of Signals and Systems, Chalmers University of Technology, 412 96, Gothenburg, Sweden (email: m.matthaiou@qub.ac.uk).}
\thanks{M. Coldrey is with Ericsson Research Gothenburg, Sweden (email: mikael.coldrey@ericsson.com).}
\thanks{E.~Bj\"ornson is with the Department of Electrical Engineering (ISY), Link\"{o}ping University, Link\"{o}ping, Sweden (email: emil.bjornson@liu.se).}
\thanks{Parts of this paper were published at the IEEE International Conference on Communications (ICC), Sydney, Australia, June 2014 \cite{xinlin2014training_ICC}.}}

%

\maketitle

\begin{abstract}

Radio-frequency (RF) impairments, which intimately exist in wireless communication systems, can severely limit the performance of multiple-input multiple-output (MIMO) systems. Although we can resort to compensation schemes to mitigate part of these impairments, a certain amount of residual impairments always persists. In this paper, we consider a training-based point-to-point MIMO system with residual transmit RF impairments (RTRI) using spatial multiplexing transmission. Specifically, we derive a new linear channel estimator for the proposed model, and show that RTRI create an estimation error floor in the high signal-to-noise ratio (SNR) regime. Moreover,  we derive closed-form expressions for the signal-to-noise-and-interference ratio (SINR) distributions, along with analytical expressions for the ergodic achievable rates of zero-forcing, maximum ratio combining, and minimum mean-squared error receivers, respectively. In addition, we optimize the ergodic achievable rates with respect to the training sequence length, and demonstrate that finite dimensional systems with RTRI generally require more training at high SNRs than those with ideal hardware. At last, we extend our analysis to large-scale MIMO configurations, and derive deterministic equivalents of the ergodic achievable rates. It is shown that, by deploying large receive antenna arrays, the extra  training requirements due to RTRI can be eliminated. In fact, with sufficiently large number of receive antennas, systems with RTRI may even need less training than systems with ideal hardware.


\end{abstract}

\vspace{-0.5cm}
\thispagestyle{empty}

\vspace{-0.5cm}
\section{Introduction}
Point-to-point multiple-input multiple-output (MIMO) systems offer wireless communication with high data rates, without requiring additional bandwidth or transmit power. The pioneering works of \cite{telatar1999capacity,foschini1998a} demonstrated a linear growth in capacity in rich scattering environments by deploying multiple antennas at both the transmitter and receiver sides. However, to fully reap the advantages that MIMO systems can offer, perfect instantaneous channel state information (CSI) is essential, especially at the receiver.

In many practical systems, a training-based (or pilot-based) transmission scheme is usually adopted to acquire CSI. In such systems, the transmitter sends a block of symbols which contain both pilot and data information. The receiver estimates the instantaneous channel realization and uses the acquired CSI to retrieve the intended data. Under ideal radio-frequency (RF) hardware assumptions, channel estimation is a  well investigated topic (e.g., \cite{hassibi2003,  biguesh2006training, mikael1, emil2010est, Love2014Training1_JSTSP, Love2014Training2_JSTSP}); however, these assumptions are quite unrealistic in practice. RF impairments, such as in-phase/quadrature-phase (I/Q) imbalance, high power amplifier non-linearities, and oscillator phase noise, are known to have a detrimental impact on practical MIMO systems \cite{schenk2008rf, studer2010residual, Jingya2014IQI_TCOM, Giuseppe2014PA_TCOM}. The influence of such individual hardware impairments is usually mitigated by using analog and digital signal processing algorithms \cite{schenk2008rf}. However, these techniques are not able to completely remove hardware impairments, such that a certain amount of residual distortions always remains.  These residual impairments stem from, for example, time-varying hardware characteristics which cannot be accurately parameterized and estimated, as well as, the randomness induced by different types of noise and imperfect compensation schemes.

The statistical properties of residual impairments have been investigated in a number of existing works.
The authors in \cite{schenk2008rf, studer2010residual}  characterized and verified experimentally that the residual distortion caused by the impairments behaves as additive and independent Gaussian noise.  This Gaussian behavior can be understood by the central limit theorem, when the distortions from many independent and different sources add up together.
In particular, for Gaussian input, one can interpret the above observation using the Bussgang Theorem \cite{Bussgang}, which states that the output process of a Gaussian input process through a non-ideal hardware is the sum of a scaled input process and a distortion process which is uncorrected with the input process \cite{Ochiai2002Clipped_TCOM, Wenyi2012Framework_TCOM}.
The very recent paper \cite{Ulf2014Impairments_GLOBECOM} has compared the above additive Gaussian model with a sophisticated hardware impairment model in massive MIMO systems, which comprised several main sources of hardware impairments. Their results  verified that the above Gaussian approximation is sufficiently accurate despite of its simplicity.

The impact of residual hardware impairments has been only scarcely investigated in few works recently.
In \cite{Emil2013imp_COML, xinlin2014capacity_ICC, Michalis2013Dhop_TCOM, Michalis2013TwoWay_COML, emil2013est}, the impact of residual RF impairments was studied under different system scenarios, with a few important conclusions being drawn. For example, \cite{Emil2013imp_COML, xinlin2014capacity_ICC} showed that impairments fundamentally limit the channel capacity in the high signal-to-noise (SNR) regime, while for low transmit power, such impact is negligible. Moreover, \cite{xinlin2014capacity_ICC} has shown that by deploying large antenna arrays at both the transmitter and receiver, the influence of impairments can be alleviated. Nonetheless, these results hold only for optimal receivers which induce high computational complexity, and are not of particular interest for larger-scale MIMO topologies. In \cite{emil2013est}, the authors claimed that the channel capacity is mainly limited by the residual impairments at the user terminals, and also proved that large-scale MIMO systems were able to tolerate larger hardware impairments than conventional small-scale MIMO systems. The authors therein also showed a few limiting effects induced by impairments in the large-antenna regime, such as estimation accuracy and capacity ceiling; however, they did not pursue any resource allocation analysis, which is of pivotal importance for the performance optimization of training-based MIMO systems.

Motivated by the above mentioned limitations, in this paper, we consider a training-based MIMO system affected by residual transmit RF impairments (RTRI). Our main goal is to investigate the impact of RTRI on the rates that are achievable with different linear receivers, as well as to find the optimal training length. More specifically, our main contributions are:
\begin{itemize}
\item We derive a linear minimum mean-squared error (LMMSE) channel estimator of the channel with optimized training matrix structure. We analytically show that RTRI introduce an irreducible estimation error floor in the high SNR regime.
\end{itemize}
\begin{itemize}
\item We derive closed-form expressions which tightly approximate the SINR distributions (outage probabilities) and ergodic achievable rates for zero-forcing (ZF), maximum ratio combining (MRC), and minimum mean-squared error (MMSE) receivers. In contrast to \cite{emil2013est}, these results hold \textit{for any finite number of antennas}. We observe achievable rate ceilings in the high SNR regime, which are analytically quantified, while in the low SNR regime, we show that the impact of RTRI vanishes. Moreover, we demonstrate that MRC receivers are much more resilient to hardware impairments than ZF and MMSE receivers.
\item We find the optimal training length of MIMO systems with RTRI. We show that, if ZF and MMSE receivers are adopted, the optimal training length can be much larger than that of systems with ideal hardware. In general, higher levels of RTRI impose more training requirements.
\item Finally, we extend our analysis to a general large-scale MIMO scenario. With the help of random matrix theory, we derive \textit{deterministic~equivalents} of the ergodic achievable rates. The optimal training length is thereafter found by solving an asymptotically optimal convex optimization problem. An interesting observation is that with large number of receive antennas, systems with RTRI may need even less training than ideal hardware systems.
\end{itemize}

The paper is structured as follows: The system model, including channel training and data transmission, is described in Section \ref{sec:SystemModel}. The estimation phase is analyzed in Section \ref{sec:estimation}. In Section \ref{sec:datatransmission}, we derive the achievable rates of different linear receivers, along with the optimal training length. We extend our analysis to large-scale MIMO systems in Section \ref{sec:MassiveMIMO}. Finally, conclusions are drawn in Section \ref{sec:Conclusion}.

\textit{Notation:} Upper and lower case boldface letters denote matrices and vectors, respectively. The trace of a matrix is expressed by $\mathrm{tr}\left(\cdot\right)$. The $n \times n$ identity matrix is represented by $\mathbf I_n$. The expectation operation is $\mathbb {E} [\cdot]$, while the matrix determinant is denoted by det$(\cdot)$. The superscripts $(\cdot)^ H$, $(\cdot)^{-1}$ and $(\cdot)^{\dagger}$ stand for Hermitian transposition, matrix inverse and pseudo-inverse, respectively. The Frobenius norm and  spectral norm are denoted by $\left\|\cdot\right\|_F$ and $\left\|\cdot\right\|_2$, respectively. The symbol $\mathcal{CN}\left(\mathbf m, \boldsymbol\Sigma\right)$ denotes a multi-variate circularly-symmetric complex Gaussian (CSCG) distribution with mean $\mathbf m$ and covariance $\boldsymbol\Sigma$. For any matrix $\mathbf H\in\mathbb C^{m\times n}$, $\mathbf h_i$ is the $i$-th column of $\mathbf H$.\vspace{-0.3cm}

\section{Signal and System Models}\label{sec:SystemModel}
We consider a block fading channel with coherence length ${T}$. During each block, the channel is constant, and is a realization of the uncorrelated Rayleigh fading model. Channel realizations between different blocks are assumed to be independent.
\vspace{-0.2cm}
\subsection{System Model With RTRI}
RF impairments  exist inherently in practical wireless communication systems. Due to these impairments, the transmitted signal will be distorted during the transmission processing,  thereby causing a mismatch between the intended signal and what is actually transmitted. Even though compensation schemes are usually adopted to mitigate the effects of these impairments, there is always some amount of residual impairments. In \cite{schenk2008rf, studer2010residual}, the authors have shown that these residual impairments on the transmit side act as additive noise. Furthermore, experimental results in \cite{studer2010residual} revealed that such RTRI behave like zero-mean complex Gaussian noise, but with the important property that their average power is proportional to the average signal power. For sufficient decoupling between different RF chains, such impairments are statistically independent across the antennas. Moreover, impairments during different channel uses are also assumed to be independent.\footnote{Residual impairments refer to the aggregate noise that is left after properly compensating the hardware-impaired transmit signals in each transmission step. Since it is conventionally assumed that the transmitted symbols across different channel uses are independent, it is also reasonable to assume that the residual impairments are also independent over each symbol. The paper [11] provides experimental justification of this assumption. More generally speaking (not restricted to RTRI), the effects of transceiver distortions, such as clipping in power amplifier nonlinearities, quantization errors, can be well-approximated as memoryless functions [16], thus they are also independent over symbols/time.} We now denote the RTRI noise as $\boldsymbol\Delta$. Then, the input-output relationship of a training-based MIMO system with $N_t$ transmit antennas and $N_r$ receive antennas, within a block of $T$ symbols, can be expressed as
\vspace{-0.3cm}
\begin{equation}
\mathbf Y = \sqrt{\frac{\rho}{N_t}}\mathbf H (\mathbf S + \boldsymbol \Delta) + \mathbf V, \label{eq:SysModelAll}
\end{equation}
where ${\mathbf S}\in\mathbb{C}^{N_t\times T}$ is the transmitted signal containing both pilot symbols and data symbols,\footnote{Note that we choose $\mathbf S$ to be full rank, i.e., all $N_t$ antennas are used to transmit signals. Although this is inefficient for the low SNR regime, it is a meaningful scheme when no CSI  at the transmitter (CSIT) is available \cite{Jorswieck2004Transmission_TSP}.} and $\mathbf H\in\mathbb C^{N_r\times N_t}$ is the channel matrix. The receiver noise and the received signal are denoted as $\mathbf V\in\mathbb{C}^{N_r\times T}$ and $\mathbf Y\in\mathbb{C}^{N_r\times T}$, respectively. Each element of $\mathbf H$ and $\mathbf V$ follows a $\mathcal{CN}(0,1)$ distribution. The average SNR at each receive antenna is denoted by $\rho$.  At last, according to the above discussion, we can characterize the RTRI noise $\boldsymbol\Delta$ according to
\vspace{-0.2cm}
\begin{align}
\boldsymbol\delta_{i} \sim \mathcal{CN}\left(\mathbf{0}, \delta^2 \mathbf I_{N_t}  \right), \mathbb E\left[\boldsymbol\delta_{i}\boldsymbol\delta_{j}^H\right] = \mathbf 0\notag\
i, j= 1, 2, \dots, T, i\neq j,
\end{align}
where ${\boldsymbol\delta}_{i}$ denotes the $i$th column of $\boldsymbol\Delta$, and the proportionality parameter $\delta$ characterizes the level of residual impairments in the transmitter. Note that $\delta$ appears in practical applications as the error vector magnitude (EVM) \cite{holma2011LTE}, which is commonly used to measure the quality of RF transceivers. For instance, 3GPP LTE has EVM requirements in the range $\left[0.08,0.175 \right]$ \cite{holma2011LTE}, where smaller EVMs are required to support high rates. The relationship between $\delta$ and EVM is defined as
\begin{equation}
\mathrm{EVM} \triangleq \sqrt{\frac{\mathbb{E}_{\boldsymbol\Delta}\left[ \left\| \boldsymbol\Delta \right\|_F^2 \right]}{\mathbb{E}_{\mathbf S}\left[ \left\| \mathbf S\right\|_F^2 \right]}} = \delta.
\end{equation}
Evidently, when $\delta = 0$, the system model simplifies to the ideal hardware case. The EVM can be measured in practice using well-established experiment setups, e.g. \cite{Agilent8Hints}.
We can now decompose the system model in (\ref{eq:SysModelAll}) into the training and data transmission phases as follows:\vspace{-0.5cm}
\subsubsection{Training Phase}
\begin{align}\label{eq:SysModelTraining}
\mathbf Y_p = \sqrt{\frac{\rho}{N_t}}\mathbf H \left(  \mathbf S_p\! +\! \boldsymbol\Delta_p \right) + \mathbf V_p, ~\mathrm{tr}\left(\mathbf S_p\mathbf S_p^H\right) = N_tT_p,
\end{align}
where $\mathbf S_p\in\mathbb C^{N_t\times T_p}$ is the deterministic matrix of training sequences sent over $T_p$ channel uses, that is known by the receiver, $\mathbf Y_p$ is the $N_r\times T_p$ received matrix. The distortion noise caused by the RTRI is characterized as
\begin{align}
{\boldsymbol\delta_p}_{i}\sim \mathcal{CN}\left(\mathbf 0, \delta^2\mathbf I_{N_t}\right), \mathbb E\left[{\boldsymbol\delta_p}_{i}{\boldsymbol\delta_p}_{j}^H\right] = \mathbf 0,      ~i,j=1,2,\dots,T_p,i\neq j,
\end{align}
where ${\boldsymbol\delta_p}_{i}$ denotes the $i$th column of $\boldsymbol\Delta_p$.
\subsubsection{Data Transmission Phase}
\begin{align}\label{eq:SysModelData}
\mathbf Y_d \!= \!\sqrt{\frac{\rho}{N_t}}\mathbf H \left( \mathbf S_d \!+\! \boldsymbol\Delta_d\right)\! +\! \mathbf V_d,
~\mathbb{E}\Big[\mathrm{tr}\left(\mathbf {S}_d^H\mathbf S_d\right)\Big]\! =\! N_t T_d,
\end{align}
where $\mathbf S_d \in \mathbb{C}^{N_t \times T_d}$ is the matrix of data symbols sent over $T_d$ channel uses, $\mathbf Y_d$ is the $N_r \times T_d$ received signal matrix. The distortion noise caused by RTRI  is characterized as
\begin{align}
{\boldsymbol\delta_d}_{i}\sim \mathcal{CN}\left(\mathbf 0, \delta^2\mathbf I_{N_t}\right), \mathbb E\left[{\boldsymbol\delta_d}_{i}{\boldsymbol\delta_d}_{j}^H\right] = \mathbf 0,~i,j=1,2,\dots,T_d,i\neq j,
\end{align}
where ${\boldsymbol\delta_d}_{i}$ denotes the $i$th column of $\boldsymbol\Delta_d$.

Note that $\mathbf S = \left[\mathbf S_p ~\mathbf S_d\right]$ and $\boldsymbol\Delta = \left[\boldsymbol\Delta_p ~\boldsymbol\Delta_d\right]$, and the conservation of coherence block length yields $T = T_p + T_d$.

\begin{remark}
We only consider the RTRI at the transmitter side since it is more dominant than its counterpart at the receiver. In practical communication systems, the received signal will firstly pass through a low noise amplifier (LNA), such that the post-processed power for the rest of the receiver circuit (e.g. demodulator, mixer) is always within a certain range, no matter how much power is actually captured by the antenna. Moreover, even if we model the receiver residual impairments in the same way as the RTRI, i.e., its variance is proportional to the received power that is received by the antenna, it is still less dominant in some practically interesting cases; for example, low SNR applications and large-scale MIMO with many receive antennas \cite{xinlin2014capacity_ICC}.
\end{remark}
\vspace{-0.3cm}
\section{Channel Estimation}\label{sec:estimation}\vspace{-0.5cm}
In this section, we analyze the impact of RTRI on the training phase. Our objective is to characterize how estimation accuracy is affected by RTRI. Conditioned on the availability of RTRI distortion distribution at the receiver, the current channel realization within each block can be estimated through the received signal $\mathbf Y_p$. The received signal contains both Gaussian distributed terms $\mathbf H\mathbf S_p$ and $\mathbf V_p$ and double Gaussian distributed term $\mathbf H\boldsymbol\Delta_p$; therefore, closed-form MMSE estimator is difficult to determine. Thus, we resort to LMMSE estimator, whose closed-form expression can be derived through tractable manipulations, and which performs better than other non-Bayesian estimators, e.g., least square (LS) estimator.


%
Given the signal model in \eqref{eq:SysModelTraining}, we can derive the LMMSE estimator by directly extending the classical estimation results \cite{biguesh2006training}, which is given by
\begin{equation}\label{eq:estimator}\vspace{-0.3cm}
\hat{\mathbf H} = \sqrt{\frac{\rho}{N_t}}\mathbf Y_p\bigg(\frac{\rho}{N_t} \mathbf S_p^H{\mathbf S_p} + \left( \delta^2\rho + 1 \right)\mathbf I_{T_p} \bigg)^{-1}\mathbf S_p^H.
\end{equation}
For any given training symbol matrix $\mathbf S_p$, \eqref{eq:estimator} minimizes the mean-squared error (MSE) from the true channel, which is defined as $\mathrm{MSE} \triangleq    \mathrm{tr}\left(\tilde{\mathbf C}\right)$, where $\tilde{\mathbf C}\!\triangleq\!\mathbb E\left[ \tilde{\mathbf H}^H\tilde{\mathbf H}\right]$ represents the estimation error covariance matrix, and $\tilde{\mathbf H} \triangleq \mathbf H - \hat{\mathbf H}$ is the estimation error.
Furthermore, using a similar approach as in \cite{biguesh2006training}, we find that the optimal training sequence which minimizes the MSE should satisfy $\mathbf S_p\mathbf S_p^H = T_p\mathbf I_{N_t}$, and the corresponding MSE reduces to
\begin{equation}\label{eq:MSE_expression}\vspace{-0.3cm}
\mathrm{MSE} = \frac{N_rN_t}{1 + \epsilon}, ~\mathrm{with}~ \epsilon \triangleq \frac{\rho T_p}{N_t(\rho\delta^2+1)}.
\end{equation}
%
The resulting estimation error covariance matrix therefore becomes $\tilde{\mathbf C} = \frac{N_r}{1+\epsilon}\mathbf I_{N_t}$. It is straightforward to show that $\tilde{\mathbf H}$ has i.i.d. entries with zero mean. The variance of each element in $\tilde{\mathbf H}$, which is also defined here as the normalized MSE (NMSE), can be expressed as $\sigma_{\tilde{\mathbf H}}^2 = \frac{1}{N_rN_t}\mathrm{tr}\left(\tilde{\mathbf C}\right) = \frac{1}{1+\epsilon}$. By the orthogonality principle of LMMSE estimates \cite{kaybook1}, each element in $\hat{\mathbf H}$ has a variance of $\sigma_{\hat{\mathbf H}}^2 = 1 - \sigma_{\tilde{\mathbf H}}^2 = \frac{\epsilon}{1+\epsilon}$. If we examine at the expression of the NMSE, it is trivial to find that the estimator reaches an error floor at high SNRs, which is quantified by taking $\rho\rightarrow\infty$ as
\begin{equation}\vspace{-0.3cm}
\mathrm{NMSE}_\mathrm{limit} = \frac{1}{1+\frac{T_p}{N_t\delta^2}}.\label{eq:MSE_limit}
\end{equation}
Obviously, the value of this floor depends on the level of RTRI; in general, large RTRI will cause severe degradation of the channel estimates. We can also see from (\ref{eq:MSE_limit}) that, for a fixed level of RTRI, an increase in the training sequence length $T_p$ decreases the MSE monotonically.

Figure \ref{fig:MSE} shows the NMSE, $\sigma_{\tilde{\mathbf H}}^2$, of a MIMO system with $N_t = N_r = 4$ for different levels of impairments. In this case, we use $T_p\!=\!4$ channel uses to transmit pilot symbols, which is the minimum length required to estimate all channel dimensions. Without the existence of RTRI, increasing the transmit power decreases the NMSE monotonically towards zero. However, in the presence of RTRI, we observe a fundamentally different behavior. Specifically, when the transmit power becomes high, impairments will generate an irreducible error floor, which was explicitly quantified in \eqref{eq:MSE_limit}.\vspace{-0.7cm}
\begin{figure}[h]
\begin{centering}
\includegraphics [height=0.35\textwidth]{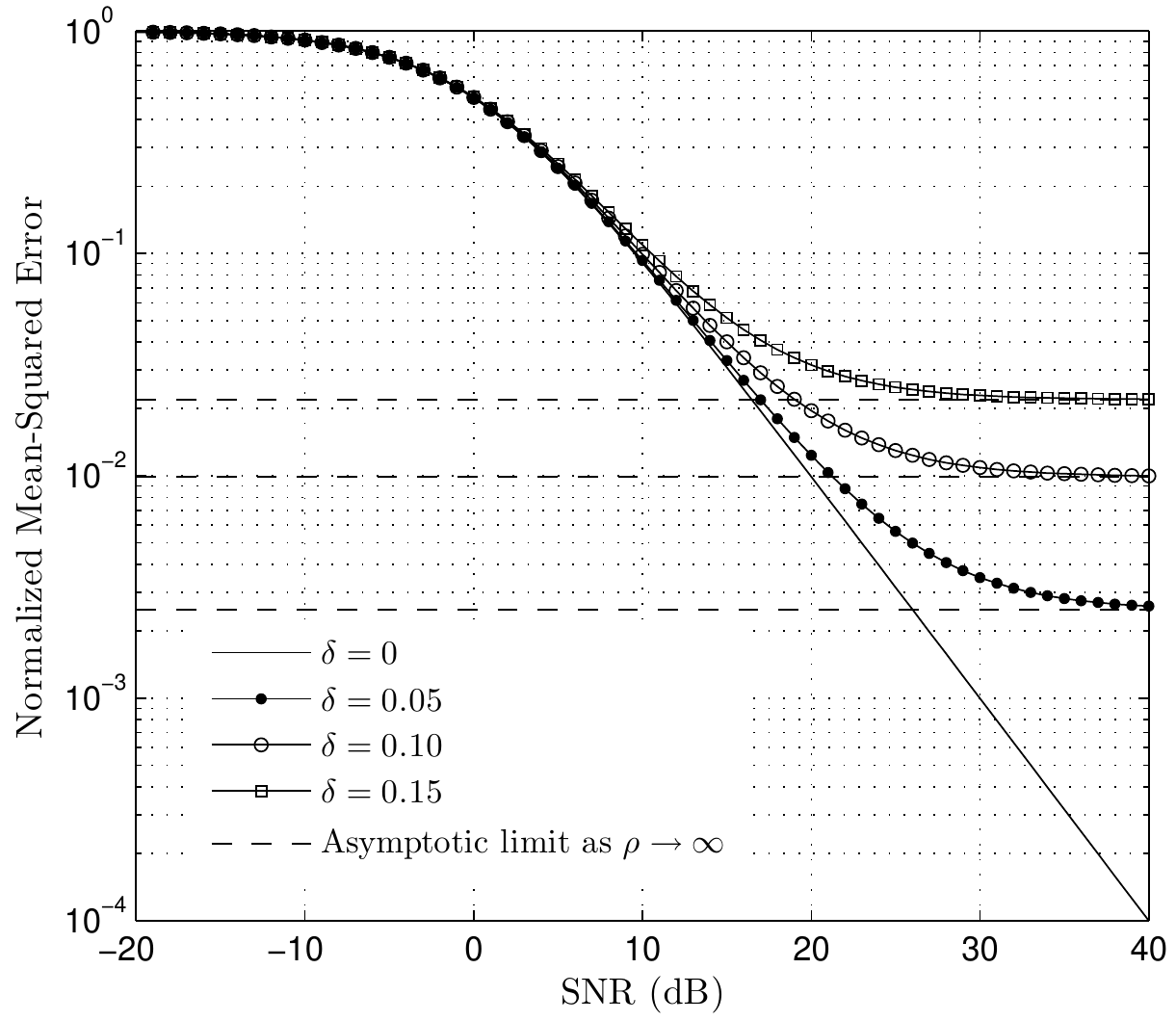}
\caption{Normalized mean-squared error (MSE) for different levels of impairments ($N_t = N_r = 4, T = 100, T_p = 4$).}\label{fig:MSE}
\end{centering}
\end{figure}
\vspace{-1.2cm}
\section{Data Transmission}\label{sec:datatransmission}
In this section, we investigate the data transmission of a training-based MIMO system with spatial multiplexing scheme. We seek to analyze the impact of RTRI on the ergodic achievable rates of three most common linear receivers, namely ZF, MRC and MMSE receivers, and to find the optimal training length that maximizes the corresponding ergodic achievable rates. Specifically, for a given coherence time $T$, we can formulate the following optimization problem:\vspace{-0.3cm}
\begin{align}\label{eq:Optimization1}
&\textrm{maximize}  ~~~R\notag\\\vspace{-0.4cm}
&\textrm{subject to}  ~~~N_t\leq T_p<T,
\end{align}
where $R$ is the ergodic achievable rate parameterized by $T_p$. We derive  SINR distributions for each type of receiver to evaluate $R$. The optimal training length $T_p^{*}$ is found through a line search over all possible $T_p$ values that satisfy the conditions in (\ref{eq:Optimization1}). In the following, we pursue our analysis for exact SNRs, as well as, high and low SNRs.\vspace{-0.7cm}
\subsection{Exact SNR analysis}\label{sec:RateExact}
We start by writing the $N_r\times 1$ received vector through (\ref{eq:SysModelData})\vspace{-0.3cm}
\begin{align}
{{\mathbf y}_d} & = \sqrt{\frac{\rho}{N_t}}{\hat {\mathbf H}}{\mathbf s_d} + \sqrt{\frac{\rho}{N_t}}\left({\tilde {\mathbf H}}{\mathbf s_d} + \mathbf H{\boldsymbol\delta_d}\right) +\mathbf v_d\\
& = \sqrt{\frac{\rho}{N_t}}\hat {\mathbf h}_k{s_d}_k +    \underbrace{ \sqrt{\frac{\rho}{N_t}}\left(\sum_{i = 1, i \neq k}^{N_t}\hat {\mathbf h}_i{s_d}_i + {\tilde {\mathbf H}}{\mathbf s_d} + \mathbf H{\boldsymbol\delta_d}\right) +\mathbf v_d}_{\triangleq{{\mathbf z}_d}_k}, ~k = 1, 2, \dots, N_t
\end{align}
where ${s_d}_k$ is the transmitted data symbol on the $k$-th transmit antenna, and ${\mathbf z_d}_k$ is the total effective noise plus interference on the $k$-th spatial stream. To recover the signal ${s_d}_k$, ${{\mathbf y}_d}$ is weighted by a $1\times N_r$ vector $\hat {\mathbf a}_k^H$, which for ZF, MRC and MMSE linear receivers is respectively given by
\begin{equation}
\hat {\mathbf a}_k^H = \begin{cases}
\hat {\mathbf g}_k^H, & \textrm{for ZF}\\
\hat{\mathbf h}^H_k, & \textrm{for MRC}\\
\hat{\mathbf h}^H_k\left(\mathbf R_{\mathbf z_{dk}}\right)^{-1} & \textrm{for MMSE}
\end{cases}
\end{equation}
where $\hat {\mathbf g}_k$ is the $k$-th column of $\hat {\mathbf G}$, with $\hat {\mathbf G} \triangleq \hat{\mathbf H}^{\dagger} = \hat{\mathbf H}\left(\hat{\mathbf H}^H\hat{\mathbf H}\right)^{-1}$, and $\mathbf R_{{\mathbf z_d}_k} \triangleq \mathbb E\left[{\mathbf z_{dk}}{\mathbf z^H_{dk}} | \hat{\mathbf H}\right]$ is the covariance matrix of ${{{\mathbf z}_d}_k}$, which is explicitly given by\vspace{-0.2cm}
\begin{equation}
\mathbf R_{{\mathbf z_d}_k} = \frac{\rho}{N_t}\sum_{i = 1, i \neq k}^{N_t}\hat {\mathbf h}_i{\hat {\mathbf h}_i}^H + \frac{\rho \delta^2}{N_t}\hat{\mathbf H}\hat{\mathbf H}^H + \left( \frac{\rho+\rho\delta^2 +1+\epsilon}{1+\epsilon} \right)\mathbf I_{N_r}.
\end{equation}
Multiplying the received signal with the weighting vector, and conditioned on that $\hat{\mathbf H}^H$ is known at the receiver, the SINR of ZF, MRC and MMSE receivers on the $k$-th spatial stream are respectively given as follows
\begin{equation}\vspace{-0.5cm}
\gamma_{k,\mathrm{ZF}} = \frac{1}{\delta^2 + c_0\left[\left(\bar {\mathbf H}^H\bar {\mathbf H}\right)^{-1}\right]_{k,k}},\label{eq:SINR_ZF}
\end{equation}
\begin{equation}
\gamma_{k,\mathrm{MRC}} = \frac{ \lvert {\bar {\mathbf h}_k}^H{\bar {\mathbf h}_k} \rvert^2 }{\sum_{i=1,i\neq k}^{N_t} \lvert {\bar {\mathbf h}_k}^H{\bar {\mathbf h}_i} \rvert^2 +\delta^2\sum_{i=1}^{N_t} \lvert {\bar {\mathbf h}_k}^H{\bar {\mathbf h}_i} \rvert^2 + c_0\lvert {\bar {\mathbf h}_k}^H{\bar {\mathbf h}_k}\rvert},\label{eq:SINR_MRC}
\end{equation}
\begin{equation}
\gamma_{k,\mathrm{MMSE}} = {\bar {\mathbf h}_k}^H\left( \sum_{i=1,i\neq k}^{N_t}{\bar {\mathbf h}_i}{\bar {\mathbf h}_i}^H + \delta^2\sum_{i=1}^{N_t}{\bar {\mathbf h}_i}{\bar {\mathbf h}_i}^H +c_0\mathbf I_{N_r}   \right)^{-1}\bar {\mathbf h}_k,\label{eq:SINR_MMSE}
\end{equation}
where we have defined $c_0 \triangleq \frac{N_t\left(\rho+\rho\delta^2+1+\epsilon\right)}{\rho \epsilon}$ for notational convenience, and $\bar {\mathbf h}_i$ denotes the $i$th column of $\bar {\mathbf H}$. We also define $\bar{\mathbf H} \triangleq \frac{1}{\sigma^2_{\hat{\mathbf H}}}\hat{\mathbf H}$, which has independent and approximately $\mathcal{CN}(0,1)$ entires.\footnote{Due to the multiplicative term $\mathbf H\boldsymbol\Delta_p$, which appears in the estimated channel $\hat{\mathbf H}$, $\bar{\mathbf H}$ is not exactly Gaussian distributed. However, we can verify easily through simulations that this it can be tightly approximated as Gaussian distributed even for  high levels of RTRI.} Note that for ZF receivers, we require $N_r\geq N_t$, whilst we consider arbitrary antenna configurations for MRC and MMSE receivers. We now give the SINR cumulative distribution functions (CDFs) for these receivers in the following proposition, which, to best of our knowledge, are new.
\begin{proposition}\label{prop:cdf}
For ZF, MRC and MMSE linear receivers, the CDF of the output SINR of an arbitrary spatial stream can be tightly approximated as\footnote{Although our analytical results are based on the approximation mentioned in Footnote 3, we henceforth use equality signs without significant notational abuse. Note that our simulation results indicate that the rate approximations remain remarkably tight for a plethora of system parameters.}
\begin{align}\label{eq:cdf_ZF}
F_{\Gamma_{\mathrm{ZF}}}(\gamma) &= \begin{cases}
1-e^{-\frac{c_0\gamma}{1-\delta^2\gamma}}\sum\limits_{k=0}^{N_r-N_t}\frac{\left({\frac{c_0\gamma}{1-\delta^2\gamma}}\right)^k}{k!}, &0\leq\gamma<\frac{1}{\delta^2},\\
1, &\gamma \geq \frac{1}{\delta^2},\\
\end{cases}\\
\label{eq:cdf_MRC}
F_{\Gamma_{\mathrm{MRC}}}(\gamma) &= \begin{cases}
1-\frac{e^{-\frac{c_0\gamma}{1-\delta^2\gamma}} }{\left(\frac{1+\gamma}{1-\delta^2\gamma}\right)^{N_t-1}}\sum\limits_{k=0}^{N_r-1}\sum\limits_{p=0}^{k}\frac{\alpha_{p,k}\left( \frac{c_0\gamma}{1-\delta^2\gamma} \right)^k}{\left( \frac{1+\gamma}{1-\delta^2\gamma} \right)^p}, &0\leq\gamma<\frac{1}{\delta^2},\\
1, &\gamma \geq \frac{1}{\delta^2},\\
\end{cases}\\
\label{eq:cdf_MMSE}
F_{\Gamma_{\mathrm{MMSE}}}(\gamma) &= \begin{cases}
1- \frac{e^{-\frac{c_0\gamma}{1-\delta^2\gamma}}}{\left(\frac{1+\gamma}{1-\delta^2\gamma}\right)^{N_t-1}}\sum\limits_{k=0}^{N_r-1}\beta_k\left(\frac{c_0\gamma}{1-\delta^2\gamma}\right)^{k}, &0\leq\gamma<\frac{1}{\delta^2},\\
1, &\gamma \geq \frac{1}{\delta^2},
\end{cases}
\end{align}
where $\alpha_{p,k} \triangleq \frac{{N_t+p-2\choose p} \left(\frac{1+\delta^2}{c_0}\right)^p }{(k-p)!}$, and $\beta_k \triangleq \sum\limits_{p = \max(0,k-N_t+1)}^{k} \frac{{N_t-1\choose k-p}\left(\frac{c_0}{1+\delta^2}\right)^{p-k}}{p!}$.
\end{proposition}
\begin{IEEEproof}
See Appendix \ref{appendix:zf_cdf}.\end{IEEEproof}

The above closed-form expressions for the SINR CDF distributions incorporate the statistical impact of RTRI, and most importantly, are valid for any finite number of transmit and receive antennas. Note that for $\delta = 0$, the above SINR expressions reduce to the ones in \cite{Mckay2009LinearReceiver} with ideal hardware assumptions.
Mathematically speaking, (\ref{eq:cdf_ZF})-(\ref{eq:cdf_MMSE}) are equivalent to the outage probabilities,\footnote{The outage probability per stream is always the lowest when transmitting only one stream, however, it renders a lower average throughput. Thus, in the following numerical analysis, we use all $N_t$ antennas for capacity-approaching transmission. Note that, for the case which only $N_t'<N_t$ antennas are used for transmission, we can evaluate the outage probability by substituting $N_t = N_t'$ into Proposition \ref{prop:cdf}. } since
\begin{equation}
P_\mathrm{out}(x) = \mathrm{Pr}\{\gamma\leq x\} = F_\Gamma(x).
\end{equation}
Note that Proposition \ref{prop:cdf} showcases the important fact that the outage probability is always 1 for $\gamma\geq \frac{1}{\delta^2}$. This implies that an SINR wall exists, which cannot be crossed by simply increasing the transmit power. The value of this wall is inversely proportional to the square of the level of impairments. Note that similar conclusions were also drawn in \cite{Michalis2013Dhop_TCOM} in the context of dual-hop relaying.

We now turn our attention to the ergodic achievable rates, and invoke the following expression of the ergodic achievable rate \cite[Lemma 1]{Mckay2009LinearReceiver}:\vspace{-0.3cm}
\begin{align}\label{eq:cdf_rate}
R &=  \frac{T_d}{T}\mathbb E\left[ N_t\log_2\left(1+\gamma\right)\right]\notag\\
&=\frac{T_dN_t}{\ln(2)T}\int_0^{\infty}\frac{1-F_{\Gamma}(\gamma)}{1+\gamma}d\gamma,
\end{align}
where $\gamma$ is found in (\ref{eq:SINR_ZF}), (\ref{eq:SINR_MRC}) and (\ref{eq:SINR_MMSE}) for ZF, MRC and MMSE receivers, respectively, and the function $F_{\Gamma}(\gamma)$ represents the CDF of $\gamma$. Specifically, the corresponding ergodic achievable rates are given in the following theorem.
\begin{theorem}\label{theorem:anal_rates} For training-based MIMO systems with RTRI, the ergodic achievable rates of ZF, MRC and MMSE receivers can be tightly approximated as
\begin{equation}
\!\!\!\!R_\mathrm{ZF} =
\frac{T_dN_t}{\ln(2)T}\sum_{k=0}^{N_r-N_t}c_0^k\left(e^{\frac{c_0}{\delta^2+1}}E_{k+1}\left(\frac{c_0}{\delta^2+1}\right)-e^{\frac{c_0}{\delta^2}}E_{k+1}\left(\frac{c_0}{\delta^2}\right)\right),\label{eq:anal_ZF_rate}
\end{equation}
\begin{align}
\!\!\!\!R_\mathrm{MRC} &= \frac{T_dN_t}{\ln(2)T}\sum_{k=0}^{N_r-1}\sum_{p=0}^{k}(-1)^{p+N_t}\alpha_{p,k}\delta^{2(p+N_t-1)}\Gamma(k+1)\notag\\
&\>\>\>\>\times\!\left(\!\!\Psi\!\left(1,-k\!+\!1;\!\frac{c_0}{\delta^2}\right) \!-\! \sum_{i=1}^{p+N_t}\!\!\left(\!-\frac{c_0}{\delta^2(\delta^2\!\!+\!\!1)}\!\right)^{p+N_t-i}\!\!\!\!\!\!\Psi\!\left(p\!+\!N_t\!-\!i\!+\!1,p\!+\!N_t\!-\!k\!-\!i\!+\!1;\frac{c_0}{\delta^2\!\!+\!\!1}\right)\!\!\right)\!,\label{eq:anal_MRC_rate}\\\vspace{-0.2cm}
\!\!\!\!R_\mathrm{MMSE} &= \frac{T_dN_t}{\ln(2)T}\sum_{k=0}^{N_r-1}(-1)^{N_t}\beta_k\delta^{2(N_t-1)}\Gamma(k+1)\notag\\
&\>\>\>\>\times\!\left(\!\Psi\!\left(1,-k\!+\!1;\frac{c_0}{\delta^2}\right)\! -\! \sum_{i=1}^{N_t}\left(-\frac{c_0}{\delta^2(\delta^2\!\!+\!\!1)}\right)^{N_t-i}\!\!\!\Psi\left(N_t\!-\!i\!+\!1,-k\!+\!N_t\!-\!i\!+\!1;\frac{c_0}{\delta^2\!+\!1}\right)\!\!\right),  \label{eq:anal_MMSE_rate}
\end{align}
where $E_n(z)$ and $\Gamma(n)$ are the exponential integral function \cite[Eq. (8.211.1)]{gradshteyn2007a} and the Gamma function \cite[Eq. (8.310.1)]{gradshteyn2007a}, respectively, while $\Psi(a,b;z)$ is the regularized hypergeometric function \cite[Eq. (9.210.2)]{gradshteyn2007a}.
\end{theorem}
\begin{proof}
See Appendix \ref{appendix:anal_rates}.
\end{proof}
These ergodic achievable rates are valid for arbitrary antenna configurations, and can be efficiently evaluated and easily implemented. Unfortunately, they offer limited physical insights, and we, therefore, now elaborate on the low and high SNR regimes.
\subsection{Low SNR analysis, $\rho\rightarrow 0$ }\label{sec:RateLow}
For some practical applications (e.g. battery-limited users or sensors), low SNR is of particular interest. Specifically, for ZF, MRC, and MMSE receivers, we have the following results.
\subsubsection{ZF receivers}
As $\rho\rightarrow 0$, the ZF SINR becomes\vspace{-0.3cm}
\begin{equation}
\gamma_{\mathrm{ZF}}^{\rho\rightarrow 0} = \frac{T_p\rho^2}{N_t^2\left[\left(\bar {\mathbf H}^H\bar {\mathbf H}\right)^{-1}\right]_{k,k}}.
\end{equation}
Consequently, we can approximate the ergodic achievable rate as below\vspace{-0.3cm}
\begin{align}
R_{\mathrm{ZF}}^{\rho\rightarrow 0} &= \mathbb E\left[\frac{T_dN_t}{T}\log_2\left( 1+\frac{T_p\rho^2}{N_t^2\left[\left(\bar {\mathbf H}^H\bar {\mathbf H}\right)^{-1}\right]_{k,k}}\right) \right]\notag\\
&\overset{(a)}\approx\frac{T_p(T-T_p)\rho^2}{\ln(2)TN_t}\mathbb E\left[ \frac{1}{\left[\left(\bar {\mathbf H}^H\bar {\mathbf H}\right)^{-1}\right]_{k,k}}\right]\notag\\
&\overset{(b)}=\frac{T_p(T-T_p)(N_r-N_t+1)}{\ln(2)TN_t}\rho^2,\label{eq:ZFLowRate}
\end{align}
where $(a)$ is obtained by using truncated Taylor's expansion, and $(b)$ is obtained by noticing that $1\slash\left[\left(\bar {\mathbf H}^H\bar {\mathbf H}\right)^{-1}\right]_{k,k}$ follows a chi-squared distribution with $2(N_r-N_t+1)$ degrees of freedom and is scaled by $1\slash 2$ \cite{winters1994diversity}.

\subsubsection{MRC receivers}
For very low transmit power, we can remove negligible terms from (\ref{eq:SINR_MRC}), and arrive at the following MRC SINR expression\vspace{-0.3cm}
\begin{equation}
\gamma_{\mathrm{MRC}}^{\rho\rightarrow 0} = \frac{T_p\rho^2{\lvert {\bar {\mathbf h}_k}^H{\bar {\mathbf h}_k} \rvert}}{N_t^2}.
\end{equation}
Through similar mathematical manipulations, we get the achievable rate of MRC receivers at low SNRs as
\begin{equation}
R_{\mathrm{MRC}}^{\rho\rightarrow 0} \approx\frac{T_p(T-T_p)\rho^2}{\ln(2)TN_t}\mathbb E\left[ \lvert {\bar {\mathbf h}_k}^H{\bar {\mathbf h}_k} \rvert\right] \overset{(a)}=\frac{T_p(T-T_p)N_r}{\ln(2)TN_t}\rho^2,\label{eq:MRCLowRate}
\end{equation}
where $(a)$ is achieved by using the fact that $\lvert {\bar {\mathbf h}_k}^H{\bar {\mathbf h}_k} \rvert$ is a complex chi-square distributed random variable with $2N_r$ degrees of freedom.

\subsubsection{MMSE receivers}

For MMSE receivers, we first rewrite (\ref{eq:SINR_MMSE}) as
\begin{align}
\gamma_{\mathrm{MMSE}} &\overset{(a)}= \frac{{\bar {\mathbf h}_k}^H\left(\sum_{i=1,i\neq k}^{N_t}{\bar {\mathbf h}_i}{\bar {\mathbf h}_i}^H + \frac{c_0}{1+\delta^2}  \mathbf I_{N_r}   \right)^{-1}{\bar {\mathbf h}_k}}{1+\delta^2+\delta^2{\bar {\mathbf h}_k}^H\left(\sum_{i=1,i\neq k}^{N_t}{\bar {\mathbf h}_i}{\bar {\mathbf h}_i}^H + \frac{c_0}{1+\delta^2}\mathbf I_{N_r}   \right)^{-1}{\bar {\mathbf h}_k} }\label{eq:MMSE_SINR_trasform1}\\
&\overset{(b)}=\frac{1-\left[ \left(\mathbf I_{N_t} + \frac{1+\delta^2}{c_0}\bar{\mathbf H}^H\bar{\mathbf H}  \right)^{-1}  \right]_{k,k}}{\delta^2 + \left[ \left(\mathbf I_{N_t} + \frac{1+\delta^2}{c_0}\bar{\mathbf H}^H\bar{\mathbf H}  \right)^{-1}  \right]_{k,k}}\label{eq:MMSE_SINR_trasform2},
\end{align}
where $(a)$ follows by applying Lemma \ref{lemma:MatrixInv} in Appendix A , while $(b)$ is obtained by noticing that the quadratic form on the numerator of (\ref{eq:MMSE_SINR_trasform1}) is a classic MMSE SINR expression with a modified noise power, hence it equals {$\frac{1}{\left[ \left(\mathbf I_{N_t} + \frac{1+\delta^2}{c_0}\bar{\mathbf H}^H\bar{\mathbf H}  \right)^{-1}  \right]_{k,k}} - 1$ }\cite{Ping2006MMSE}.

As $\rho\rightarrow 0$, we can easily find that the SINR and the ergodic achievable rate approach the following limits, respectively
\begin{align}
\gamma_{\mathrm{MMSE}}^{\rho\rightarrow 0} &= \frac{T_p\rho^2{\lvert {\bar {\mathbf h}_k}^H{\bar {\mathbf h}_k} \rvert}}{N_t^2},\\
R_{\mathrm{MMSE}}^{\rho\rightarrow 0} &=\frac{T_p(T-T_p)N_r}{\ln(2)TN_t}\rho^2.\label{eq:MMSELowRate}
\end{align}
Comparing the ergodic achievable rates of different linear receivers in (\ref{eq:ZFLowRate}), (\ref{eq:MRCLowRate}) and (\ref{eq:MMSELowRate}), we find that the achievable rates have no dependence on $\delta^2$, i.e., the impact of RTRI has vanished in the low transmit power regime. Moreover,  it can also be easily conjectured that the optimal training length always equals half of the coherence length, i.e., $T_p^* = T\slash 2$, such that as much as half of the coherence block should be devoted to channel training, which is in line with the result in \cite{hassibi2003}. More importantly, to maintain a certain level of performance, we can increase the number of receive antennas, while reducing the transmit power as fast as $\sqrt{N_r}$ and $\sqrt{N_r-N_t+1}$, for MRC (or MMSE) and ZF receivers, respectively. We can also achieve this power scaling by using less transmit antennas. Finally, as one would expect, MMSE and MRC receivers are equally good in the low SNR regime, while ZF receiver performs badly.

\subsection{High SNR analysis, $\rho\rightarrow\infty$}\label{sec:RateHigh}
From the expressions in (\ref{eq:anal_ZF_rate})-(\ref{eq:anal_MMSE_rate}), we can see that for a certain level of RTRI and system setup, the ergodic achievable rates only depend on the transmit power via the term $c_0$. As $\rho\rightarrow\infty$, $c_0\rightarrow\bar{c}_0 \triangleq \frac{\delta^2(1+\delta^2)N_t^2}{T_p}$. Therefore, the ergodic achievable rates of these linear receivers saturate in the high SNR regime. We can easily quantify these rate limits by substituting $c_0 = \bar{c}_0$ back into the corresponding achievable rate expressions in (\ref{eq:anal_ZF_rate})-(\ref{eq:anal_MMSE_rate}).

%

\subsection{Numerical Illustrations}
In Fig. \ref{fig:CDFs}, we plot the outage probabilities of linear receivers for different values of the SINR threshold. We consider two antenna configurations, i.e., $N_t = N_r = 5$ and $N_t = 5, N_r = 30$, at $\rho = 30$dB. We observe that for a certain SINR threshold, the outage probability of a hardware-impaired system is systematically higher than that of an ideal-hardware system. More importantly, RTRI create a SINR wall, which cannot be crossed by increasing the transmit power. On the other hand, for the ideal hardware case, the outage probability converges smoothly to 1. However, increasing the number of receive antennas can make the outage probability approach the limit. When comparing these receivers, we also find that MRC is much more resilient to hardware impairments than ZF and MMSE receivers, while ZF receivers are the most sensitive to hardware impairments. For instance, for $N_r = N_t = 5$, the outage probability $10^{-2}$ of ZF receivers exhibits a power penalty of $11$dB, whilst for MMSE and MRC receivers this value is $4$dB and $0$dB, respectively.
\begin{figure}[h]
\begin{centering}
\includegraphics [height=0.5\textwidth]{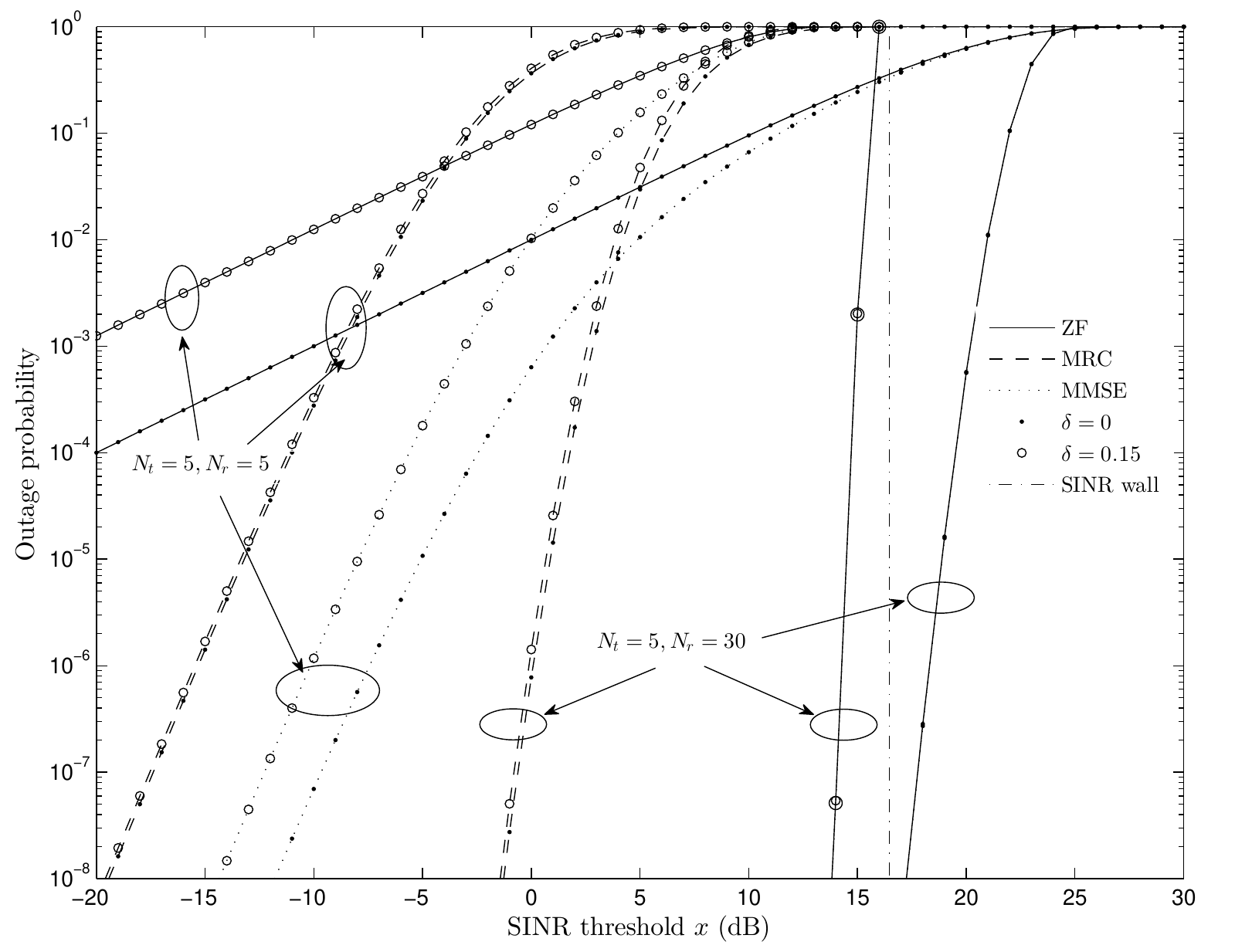}
\caption{Outage probability of linear receivers for different SINR threshold $x$.}\label{fig:CDFs}
\end{centering}
\end{figure}

Figure \ref{fig:Rate_LinRec_All} illustrates the ergodic achievable rates with ZF, MRC and MMSE linear receivers and optimized $T_p^*$.
The information-theoretic capacity, which can be approached by optimal receivers (OR), are generated according to our previous work \cite{xinlin2014training_ICC}. These curves provide upper bounds on the achievable rates obtained by linear receivers. The OR structure can be a maximum-likelihood receiver with joint decoding, or a MMSE receiver with successive interference cancellation (SIC). These curves are used as performance benchmarks.
In particular, we consider a MIMO system with $N_t = N_r = 4$ and coherence block $T=200$. The theoretical curves are plotted according to the ergodic achievable rate expressions in (\ref{eq:anal_ZF_rate})-(\ref{eq:anal_MMSE_rate}), while the numerical curves are plotted based on Monte-Carlo simulations. Clearly, there is an exact agreement between analytical and numerical results, thereby validating the correctness of the analytical expressions. At low SNRs, RTRI have negligible impact on the achievable rates for all receivers considered herein, while at high SNRs, RTRI decrease the achievable rates of ZF and MMSE receivers and generate finite achievable rate ceilings; however, RTRI do not have substantial impact on the performance of MRC receivers in this regime, due to the dominating inter-stream interference. As anticipated, MMSE receivers perform the best across the whole SNR regime. In contrast to ideal hardware systems, we can see that the performance gap between optimal receivers and MMSE/ZF receivers in hardware-impaired systems diminishes as SNR increases. In other words, for hardware-impaired systems, MMSE and ZF receivers become near optimal in the high SNR regime.

With the same system setup, we demonstrate the optimal training length $T_p^*$ for systems with different linear receivers in Fig. \ref{fig:Tp_LinRec_All}, $T_p^*$ values for optimal receivers are drawn according to \cite[Section IV-C]{xinlin2014training_ICC}.
For low SNR values, considerable amount of training is needed. As we have proved in Section \ref{sec:RateLow}, $T_p^*$ converges to $\frac{T}{2}$, as $\rho\rightarrow0$. As the SNR increases, ideal hardware systems require less channel training, until the pilot length reaches the minimum possible value $T_p^*=N_t$. However, for systems with RTRI, the conclusions are different. If ZF or MMSE receivers are considered, it is noteworthy that more pilot symbols are needed in the high SNR regime; this phenomenon indicates that, although the effective SINR  saturates due to RTRI, increasing the training length can still provide performance gains by improving the estimation accuracy. If MRC receivers are considered, there is no such extra training benefit. This is due to the fact that MRC receivers are dominated by inter-stream interference for both ideal and hardware-impaired systems, such that, the RTRI distortions have only insignificant impact on the system performance. As a result, the optimal training length for hardware-impaired systems  does not deviate much from that of ideal hardware systems. Moreover, MMSE and ZF receivers generally require more training, while MRC receivers require less training, in comparison to optimal receivers.
\begin{figure}[h]
\begin{centering}
\includegraphics [height=0.5\textwidth]{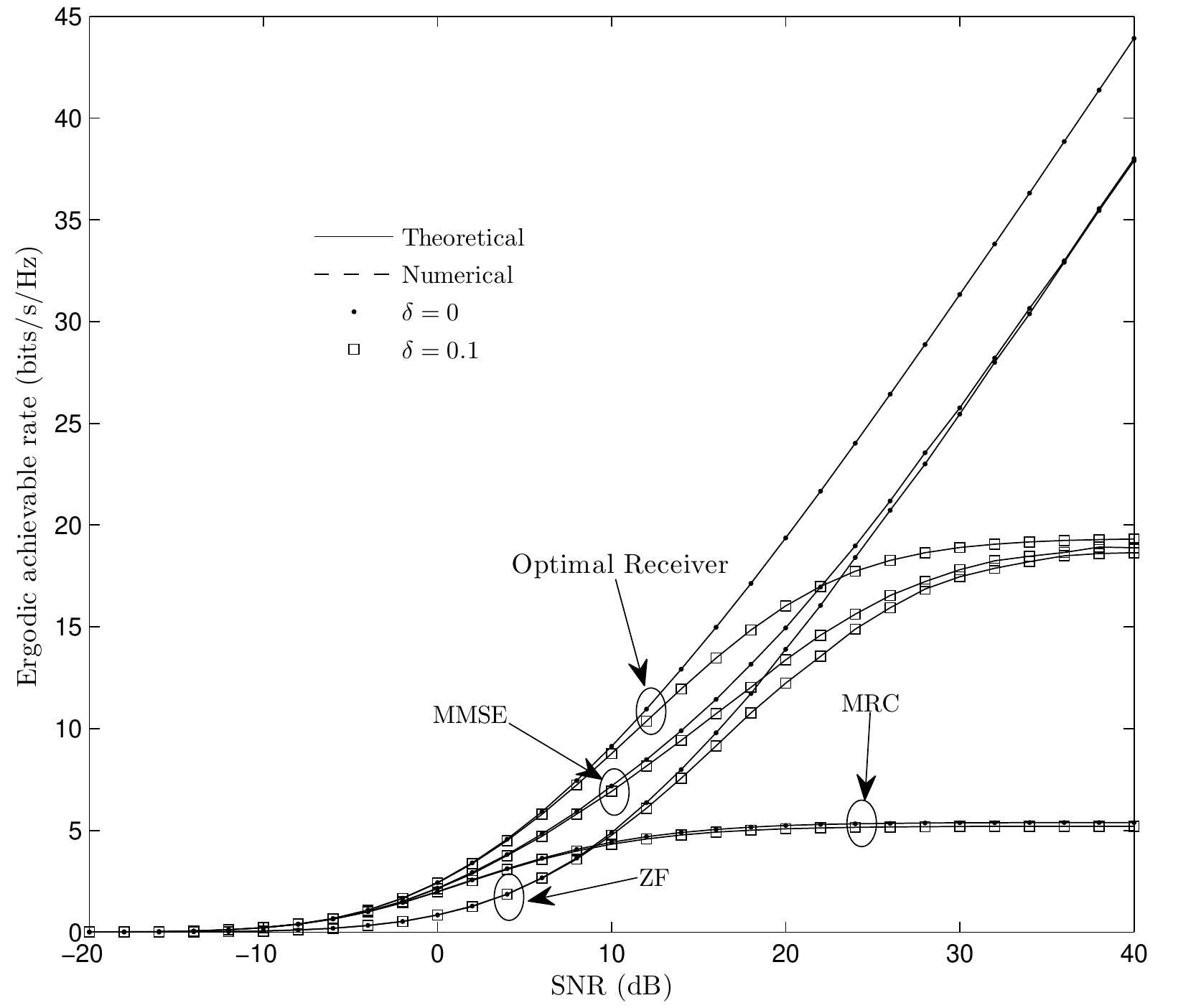}
\caption{Ergodic achievable rates with different linear receivers and optimized $T_p$ ($N_t = N_r = 4, T = 200$).}\label{fig:Rate_LinRec_All}
\end{centering}
\end{figure}
\begin{figure}[h]
\begin{centering}
\includegraphics [height=0.5\textwidth]{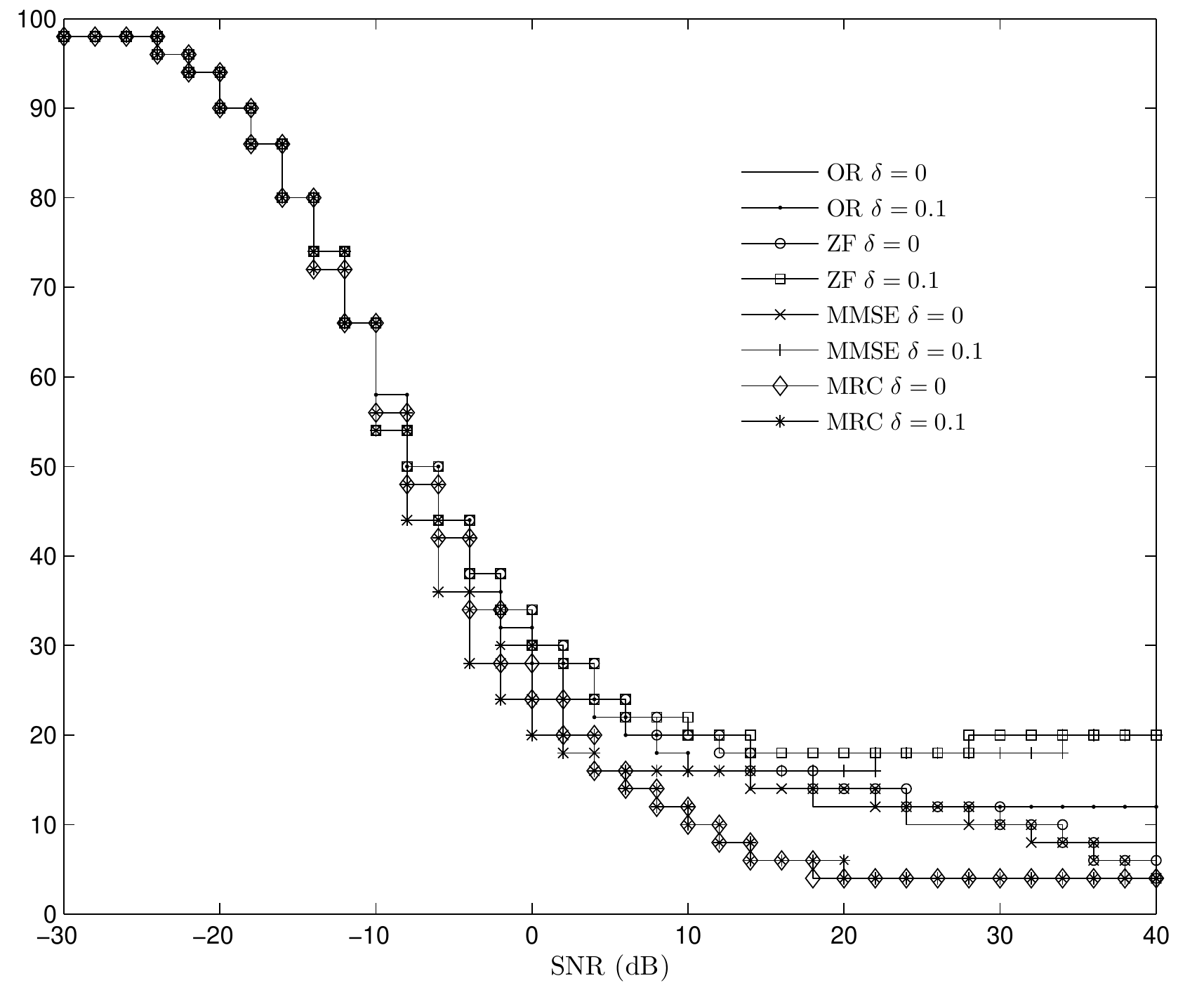}
\caption{Optimal training sequence length for linear receivers with various levels of RTRI ($N_t = N_r = 4, T = 200$).}\label{fig:Tp_LinRec_All}
\end{centering}
\end{figure}
\vspace{-0.3cm}
\section{Extension to Large-Scale MIMO Systems}\label{sec:MassiveMIMO}
Large-scale MIMO systems are considered as a key enabler for next generation wireless systems, since it can offer huge improvements in throughput, as well as, energy efficiency. Moreover, they are very robust against small-scale fading and intentional jamming \cite{Marzetta2010Massive_TWCOM}.  For large-scale MIMO systems with ideal hardware, the training overhead has been analyzed in, for example, \cite{Love2014Training1_JSTSP, Love2014Training2_JSTSP}, where it has been shown that the training requirement does not necessarily to be increased linearly with the number of antennas. However, large-scale MIMO systems are expected to deploy lower-cost radio frequency (RF) components, which are more prone to hardware impairments. Motivated by this important observation, in this section, we will investigate the impact of RTRI on large-scale systems. We derive deterministic equivalents of the ergodic achievable rates, and thereafter find the optimal training length. We assume that the coherence time $T$ is big enough to support our large dimensional analysis. As we will show later, our deterministic equivalents are accurate even for small number of antennas; therefore the dependence on large $T$ is not essential \cite{Hoydis2013Massive}.
\vspace{-0.5cm}
\subsection{Deterministic Equivalents}
For each type of linear receiver, we now present the deterministic equivalent approximations of the SINR and achievable rates in the large system limit, i.e., for $N_r, N_t \rightarrow\infty$ with a finite ratio $\beta \triangleq \frac{N_r}{N_t} \in [1, \infty)$.\footnote{Since the SINRs of each stream are statistically independent, we drop the dependence on $k$ for the sake of clarity.}
\begin{theorem}\label{theorem:bothlarge}
As $N_r, N_t \rightarrow\infty$, with a finite ratio $\beta = \frac{N_r}{N_t}\in[1, \infty)$,\footnote{For ZF receivers, we restrict $\beta$ to be in the range $(1, \infty)$. For $\beta =1$, one can show that $\bar{\mathbf H}^H\bar{\mathbf H}$ has at least one eigenvalue converging to zero \cite{Bai1993Smallest}, which makes the ZF analysis intractable. Due to page limit, we ignore the case of $\beta = 1$, one can find a detailed study of this case, for example, in \cite{Eldar2003ZFAsym_TIT}.} the SINRs of ZF, MRC, and MMSE receivers, almost surely converge to the following deterministic equivalents, respectively
\begin{align}
\gamma_{\mathrm{ZF}}^{N_r,N_t\rightarrow\infty} &\xrightarrow{a.s.} \bar{\gamma}_{\mathrm{ZF}} \triangleq \frac{\beta-1}{\delta^2(\beta-1)+c_1}\label{eq:deter_ZF},\\
\gamma_{\mathrm{MRC}}^{N_r,N_t\rightarrow\infty} &\xrightarrow{a.s.} \bar{\gamma}_{\mathrm{MRC}} \triangleq \frac{\beta}{1+\delta^2 + c_1+\delta^2\beta} \label{eq:deter_MRC},\\
\gamma_{\mathrm{MMSE}}^{N_r,N_t\rightarrow\infty} &\xrightarrow{a.s.} \bar{\gamma}_{\mathrm{MMSE}} \triangleq \frac{\frac{1}{2c_1}\left(-d + \sqrt{d^2 + \frac{4\beta c_1}{1+\delta^2}}  \right)}{1+\delta^2 + \frac{\delta^2}{2c_1}\left(-d + \sqrt{d^2 + \frac{4\beta c_1}{1+\delta^2}}  \right)}\label{eq:deter_MMSE},
\end{align}
where $c_1 \triangleq \frac{\rho+\rho\delta^2+1+\bar{\epsilon}}{\rho\bar{\epsilon}}$, $\bar{\epsilon} \triangleq \frac{T_p\rho}{N_t(\rho\delta^2+1)}$, and $d\triangleq\frac{c_1}{1+\delta^2}+1-\beta$.
\end{theorem}
\begin{IEEEproof}
See Appendix \ref{appendix:bothlarge}.
\end{IEEEproof}

Theorem \ref{theorem:bothlarge} provides us with deterministic approximations of the SINRs in the large-antenna regime, which do not rely on random channel realizations, as long as the ratio $\beta$ is fixed. The above results are quite general and tractable, and as we will show later, are very accurate even for small number of antennas. Note that some practical interesting system setups, for instance, $N_t \ll N_r$ (usually appearing in the multi-user MIMO uplink), can be treated as a special case of Theorem \ref{theorem:bothlarge}. In particular, we will give this result in the following corollary.
\begin{corollary}\label{corollary:Nr_Large}
For the case $N_t\ll N_r$, such that $\beta\rightarrow\infty$, the SINRs of ZF, MRC, and MMSE receivers, converge almost surely to the same deterministic equivalent\begin{equation}\label{eq:Nr_Large}
\gamma^{\beta\rightarrow\infty} \stackrel{a.s.}\longrightarrow \bar{\gamma} = \frac{1}{\delta^2}.
\end{equation}
\end{corollary}
This corollary is consistent with the asymptotic results in \cite{Emil2013imp_COML}.
From the above equation, we can see that the SINR reaches a limit when $N_t\ll N_r$. This limit, which is only quantified by the level of RTRI, is consistent with the outage probability analysis in (\ref{eq:cdf_ZF})-(\ref{eq:cdf_MMSE}). In this case, increasing the transmit power does not provide any SINR gain.

After establishing the asymptotic SINRs, we proceed to derive the achievable rates.
Using the continuous mapping theorem  \cite[Theorem 2.3]{Van2000Asymptotic} and dominated convergence theorem \cite[Theorem 16.4]{Billingsley2008Probability}, we can show that the achievable rates by different receivers converge to the following deterministic equivalent
\begin{equation}\label{eq:Rate_Asymp}
R^{N_r, N_t\rightarrow\infty} \longrightarrow \bar{R}= \left(1-\frac{T_p}{T}\right)N_t\log_2\left(1+\bar{\gamma}\right),
\end{equation}
where $\bar{\gamma}$ is found in (\ref{eq:deter_ZF}), (\ref{eq:deter_MRC}) and (\ref{eq:deter_MMSE}) for ZF, MRC and MMSE, respectively.
Note that one can easily prove that $\bar{R}$ is a strictly concave function in $T_p$ \cite{Boyd2004Convex}. Then, as proved in \cite[Theorem 4]{Jakob2011Training_TSP}, optimizing $R^{N_r, N_t\rightarrow\infty}$ is asymptotically identical to optimizing $\bar{R}$. Therefore, solving the following optimization problem will return the desired $T_p^*$:
\begin{align}\label{eq:Optimization}
&\textrm{maximize}  ~~~\bar{R}^{N_r, N_t\rightarrow\infty}\notag\\
&\textrm{subject to}  ~~~N_t\leq T_p<T.
\end{align}
With the results in Theorem \ref{theorem:bothlarge} and (\ref{eq:Rate_Asymp}), we can find $T_p^*$, for example, via a line search over different $T_p$ values (e.g., using the bisection method). As for the special case in Corollary \ref{corollary:Nr_Large}, it follows readily that the achievable rates of different receivers all converge to
\begin{equation}
R^{\beta\rightarrow\infty} \longrightarrow \bar{R} = \left(1-\frac{T_p}{T}\right)N_t\log_2\left(1+\frac{1}{\delta^2}\right).
\end{equation}
Obviously, in this scenario, $T_p^* = N_t$ channel uses should be devoted for training.  Furthermore, the number of transmit antennas should not be larger than half of the coherence time, i.e., $N_t\leq \frac{T}{2}$, otherwise, the channel training overhead  will cause a decrease in the achievable rate.

\vspace{-0.2cm}
\subsection{Numerical Illustrations}
Figure \ref{fig:DeviationAsymp} compares the deterministic equivalents of the ergodic achievable rates  with the ergodic achievable rates we derived in the previous section . For MMSE receivers, we see an excellent match in all cases, even for low dimensional systems. For small number of antennas, the achievable rate approximations of ZF receivers at low SNRs and of MRC receivers at high SNRs have moderate deviations from the true values, while these deviations vanish quickly as $N_r$ increases. Comparing these four plots, we can conclude that, for high SNR values, RTRI always make the convergence to the asymptotic expressions faster, whilst for low SNRs, RTRI slow down the convergence. We can also see that, increasing the ratio $\beta$ always guarantees better convergence.

In Fig. \ref{fig:MassiveTp}, we illustrate the optimal training length $T_p^*$ in the large-antenna regime. We note that in the upper subplot, where the system is ``not so massive'', RTRI still create extra training requirements in the high SNR regime. However, when the system is equipped with large receive antenna arrays, the optimal training length is reduced, and the optimal $T_p^*$ may be even less than those of  ideal hardware systems. The reason behind this phenomenon is that the improvement provided by estimation saturates very fast, and we can still have certain gain by allocating more time for data transmission.
\begin{figure}[h]
\begin{centering}
\includegraphics [height=0.75\textwidth]{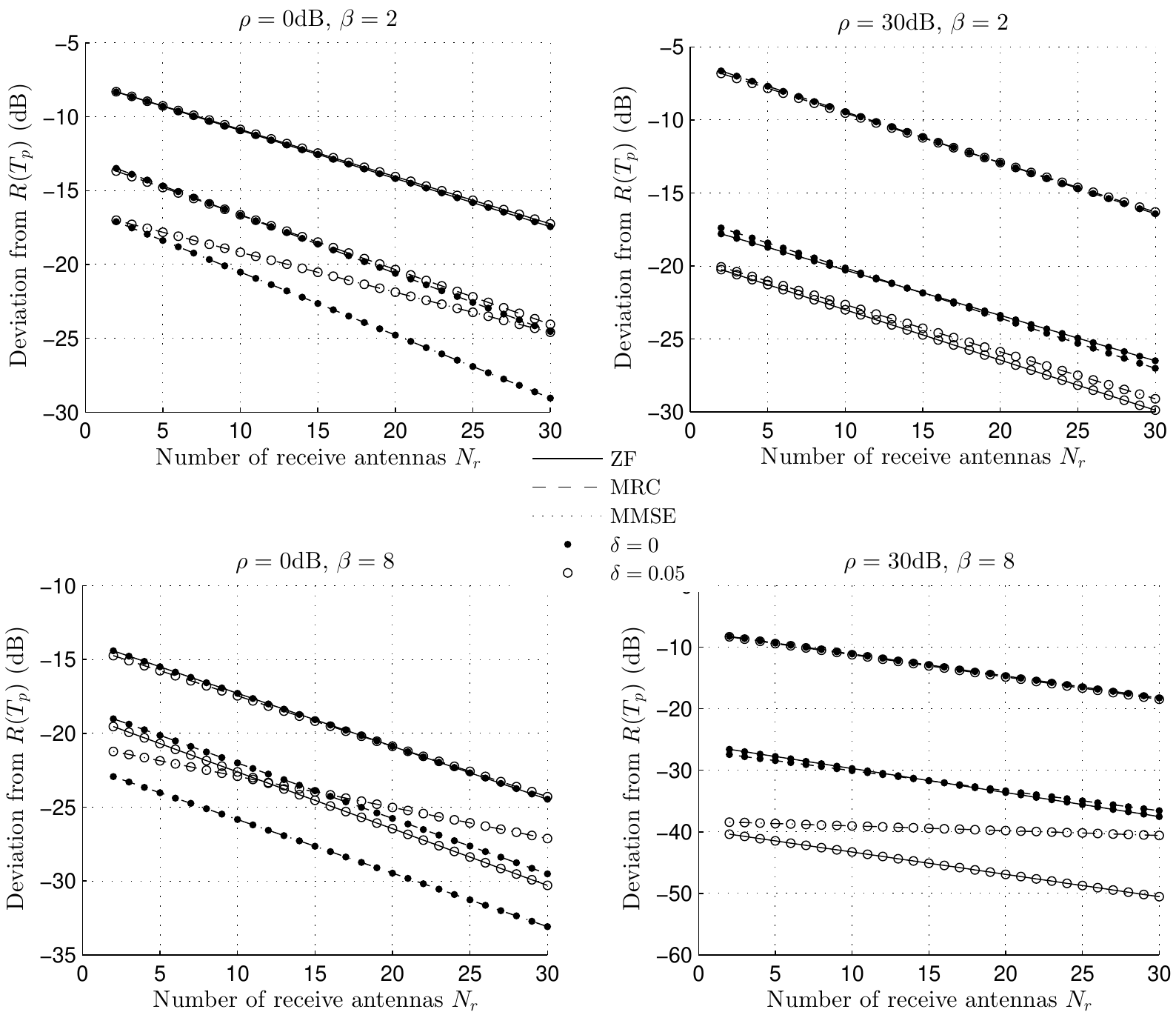}
\caption{Deviation of the achievable rate deterministic equivalents from the ergodic achievable rates ($T = 500$).}\label{fig:DeviationAsymp}
\end{centering}
\end{figure}
\begin{figure}[h]
\begin{centering}
\includegraphics [height=0.55\textwidth]{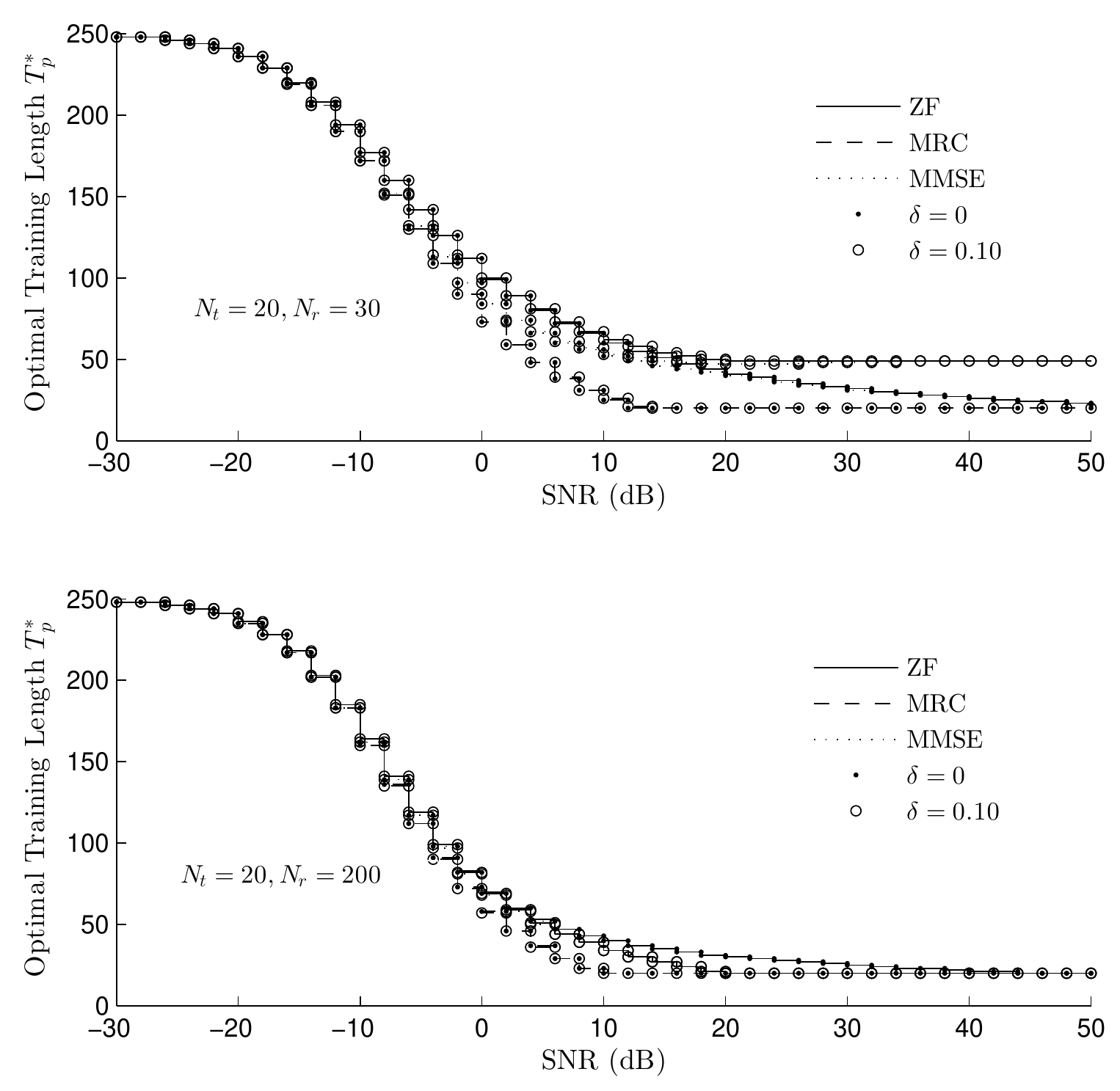}
\caption{Optimal training length in the asymptotic regime ($T = 500$).}\label{fig:MassiveTp}
\end{centering}
\end{figure}
\section{Conclusions}\label{sec:Conclusion}
In this paper, we analyzed the performance of training-based MIMO systems with residual hardware impairments. We derived a  LMMSE channel estimator for the proposed model, and deduced the irreducible estimation error floor in the high SNR regime. We derived SINR distributions for ZF, MRC and MMSE linear receivers, and found the corresponding closed-form ergodic achievable rates. We observed that the achievable rates saturated in the high SNR regime. The optimal training lengths were thereafter found through optimizing the ergodic achievable rates. Generally, RTRI imposed more training requirement at high SNRs for MMSE and ZF receivers. However, for MRC receivers, there was no difference in the optimal training length as compared to ideal hardware systems. Moreover, for large-scale systems, we derived deterministic equivalents of the ergodic achievable rates, which were shown to be accurate even for small number of antennas. In this case, we found that by deploying large receive antenna arrays, the extra training requirement imposed by RTRI was alleviated; with sufficiently large number of receive antennas, systems with RTRI needed even less training than ideal hardware systems. The reason behind this phenomenon is that, the improvement offered by channel estimation saturates very fast, while we can still have performance gain by allocating more channel uses to data transmission. Thus, it is important to find a good tradeoff between channel estimation and data transmission.


\appendices
\section{Useful random matrix theory results}
\begin{lemma}\label{lemma:MatrixInv}{\textit{(Matrix inversion lemma \cite{Hoydis2013Massive})}}
Let $\mathbf A\in\mathbb C^{N\times N}$ be a Hermitian invertible matrix, and let $\mathbf x\in\mathbb C^{N\times 1}$ and $\tau$ be arbitrary vector and scalar, respectively. Then, we have
\begin{equation}\label{eq:matrixinversion}
\mathbf x(\mathbf A + \tau\mathbf x\mathbf x^H)^{-1} = \frac{\mathbf x^H\mathbf A^{-1}}{1 + \tau\mathbf x^H\mathbf A^{-1}\mathbf x},
\end{equation}
\end{lemma}
\begin{lemma}\label{lemma:TraceLemma}{\textit{(Trace lemma \cite[Lemma B.26]{Bai2009RMBook})}}
Let $\mathbf A\in\mathbb C^{N\times N}$ and $\mathbf x$, $\mathbf y \sim \mathcal{CN}(\mathbf 0, \frac{1}{N}\mathbf I_N)$. Assume that $\mathbf A$ has uniformly bounded spectral norm with respect to $N$, and that $\mathbf x$ and $\mathbf y$ are mutually independent, and independent of $\mathbf A$. Then, as $N\rightarrow\infty$, we have
\begin{align}
\mathbf x^H\mathbf A\mathbf x-\frac{1}{N}\mathrm{tr}\left(\mathbf A\right)\stackrel{a.s.}\longrightarrow 0\\
\mathbf x^H\mathbf A\mathbf y \stackrel{a.s.}\longrightarrow 0.
\end{align}
\end{lemma}
\begin{lemma}\label{lemma:Rank1Lemma}{\textit{(Rank-1 perturbation lemma \cite{Bai1995Empirical}) }}
Let $z < 0$, $\mathbf A, \mathbf B\in\mathbb C^{N\times N}$ with $\mathbf B$ Hermitian nonnegative definite, and $\mathbf v\in\mathbb{C}^N$, then
\begin{equation}
\left| \mathrm{tr}\left((\mathbf B - z\mathbf I_N)^{-1} - (\mathbf B + \mathbf v\mathbf v^H - z\mathbf I_N)^{-1}\right)\mathbf A\right|\leq\frac{||\mathbf A||_2}{|z|}.
\end{equation}
\end{lemma}
\begin{lemma}\label{lemma:STLemma}
Let $\mathbf A\in\mathbb C^{N\times n}$, denote $F^{\mathbf A^H\mathbf A}$ as the eigenvalue CDF of $\mathbf A^H\mathbf A$, and $m_{F^{\mathbf A^H\mathbf A}}$ as the Stieltjes transform of $F^{\mathbf A^H\mathbf A}$ . Then, for $z\in\mathbb C\backslash \mathbb R^+$\footnote{The notation $z\in\mathbb C\backslash \mathbb R^+$ denotes all $z\in\mathbb C, z \not\in\mathbb R^+$.}, we have \cite[Lemma 3.1]{Couillet2011a}
\begin{equation}
\frac{n}{N}m_{F^{\mathbf A^H\mathbf A}}(z) = m_{F^{\mathbf A\mathbf A^H}}(z) + \frac{N-n}{N}\frac{1}{z}.
\end{equation}
\end{lemma}

\section{CDF of the SINR with ZF receivers}\label{appendix:zf_cdf}
Denoting $x \triangleq 1/{\left[\left(\mathbf H^H \mathbf H\right)^{-1}\right]_{k,k}}$, we know  from \cite{winters1994diversity} that $x$ follows a chi-square distribution with $2(N_r-N_t+1)$ degrees of freedom and is scaled by $1\slash2$, with its CDF given by
\begin{equation}
F_X(x) = 1-e^{-x}\sum\limits_{k=0}^{N_r-N_t}\frac{x^k}{k!}.
\end{equation}
Rewrite the SINR of ZF receiver as
\begin{equation}
\gamma_{\mathrm{ZF}} = \frac{X}{\delta^2X+c_0},
\end{equation}
then we finally have
\begin{equation}
F_{\Gamma_{\mathrm{ZF}}}(\gamma) = \Pr(\Gamma_{\mathrm{ZF}}\leq\gamma) = F_X\left(\frac{c_0\gamma}{1-\delta^2\gamma}\right).
\end{equation}

\section*{CDF of the SINR with MRC receivers}\label{appendix:mrc_cdf}
First, we rewrite $\gamma_{\mathrm{MRC}}$ as
\begin{align}
\gamma_{\mathrm{MRC}} &= \frac{\bar{\mathbf h}_k^H\bar{\mathbf h}_k}{(1+\delta^2)\sum\limits_{i=1,i\neq N_t}^{N_t}\frac{|\bar{\mathbf h}_k^H\bar{\mathbf h}_i|^2}{\bar{\mathbf h}_k^H\bar{\mathbf h}_k} + \delta^2\bar{\mathbf h}_k^H\bar{\mathbf h}_k +c_0}\notag\\
& = \frac{Z}{(1+\delta^2)Y+\delta^2Z+c_0},
\end{align}
where $Z$  is chi-square distributed with $2N_r$ degrees of freedom, and $Y$ is chi-square distributed with $2(N_t-1)$ degrees of freedom with PDF \cite{Shah1998MRC}
\begin{equation}
f_Y(y) = \frac{e^{-y}y^{N_t-2}}{(N_t-2)!}.
\end{equation}
We can now express the CDF of the SINR with MRC as
\begin{align}
F_{\Gamma_{\mathrm{MRC}}}(\gamma) &= \Pr(\Gamma_{\mathrm{MRC}}\leq\gamma)\\
&=\int_0^{\infty}F_Z\left(\frac{(1+\delta^2)\gamma}{1-\delta^2\gamma}y+\frac{c_0\gamma}{1-\delta^2\gamma}\right)f_Y(y)dy\notag\\
& = 1 - \frac{e^{-\frac{c_0\gamma}{1-\delta^2\gamma}}}{(N_t-2)!}\sum\limits_{k=0}^{N_r-1}\frac{\left(\frac{c_0\gamma}{1-\delta^2\gamma}\right)^k}{k!}\sum\limits_{p=0}^{k}{k\choose p}\left(\frac{1+\delta^2}{c_0}\right)^p\int_{0}^{\infty}y^{N_t+p-2}e^{-\frac{1+\gamma}{1-\delta^2\gamma}y}dy\label{eq:MRC_CDF_int}\\
& = 1-\frac{e^{-\frac{c_0\gamma}{1-\delta^2\gamma}}}{(N_t-2)!}\sum\limits_{k=0}^{N_r-1}\frac{\left(\frac{c_0\gamma}{1-\delta^2\gamma}\right)^k}{k!}\sum\limits_{p=0}^{k}{k\choose p}\left(\frac{1+\delta^2}{c_0}\right)^p\frac{(N_t+p-2)!}{\left(\frac{1+\gamma}{1-\delta^2\gamma}\right)^{N_t+p-1}}\\
& = 1-\frac{e^{-\frac{c_0\gamma}{1-\delta^2\gamma}}}{\left(\frac{1+\gamma}{1-\delta^2\gamma}\right)^{N_t-1}}\sum\limits_{k=0}^{N_r-1}\sum\limits_{p=0}^{k}\underbrace{\frac{{{N_t+p-2}\choose p}}{(k-p)!}\left(\frac{1+\delta^2}{c_0}\right)^p}_{\triangleq{\alpha_{p,k}}}\frac{\left(\frac{c_0\gamma}{1-\delta^2\gamma}\right)^k}{\left(\frac{1+\gamma}{1-\delta^2\gamma}\right)^{p}}.
\end{align}
The integral in \ref{eq:MRC_CDF_int} can be evaluated via the identity  \cite[Eq.~(3.381.4)]{gradshteyn2007a}.

\section*{CDF of the SINR with MMSE receivers}\label{appendix:mmse_cdf}
We start from the equivalent form in (\ref{eq:MMSE_SINR_trasform1}), and have
\begin{equation}
\gamma_{\mathrm{MMSE}} = \frac{W}{1+\delta^2+\delta^2W},
\end{equation}
where $W \triangleq {{\bar {\mathbf h}_k}^H\left(\sum_{i=1,i\neq k}^{N_t}{\bar {\mathbf h}_i}{\bar {\mathbf h}_i}^H + \frac{c_0}{1+\delta^2}  \mathbf I_{N_r}   \right)^{-1}{\bar {\mathbf h}_k}}$ is a classical MMSE SINR expression, with CDF given by \cite[Eq. 11]{Gao1998MMSE} and  \cite[Eq. 8]{Mckay2009LinearReceiver}
\begin{equation}
F_W(w) = 1- e^{-\frac{c_0x}{1+\delta^2}}\sum\limits_{k=0}^{N_r-1}\frac{A_k(w)}{k!}\left(\frac{c_0w}{1+\delta^2}\right)^k,
\end{equation}
where
$$A_k(w) \triangleq \left\{ \begin{array}{rl}
 &1, ~~~~~~~~~\mbox{ if $N_r\geq N_t+k$}, \\
  &\frac{1+\sum\limits_{i=1}^{N_r-k-1}{N_t-1\choose i}x^i}{\left(1+x\right)^{N_t-1}}, \mbox{ otherwise.}
\end{array} \right.$$
Using a similar approach as in the above proofs, followed by some simple manipulations, we obtain the final result.

\section{Proof of Theorem \ref{theorem:anal_rates}}\label{appendix:anal_rates}
\subsection{Achievable rate of ZF receivers}
Firstly, we can easily prove that the SINR should be in the range $(0,\frac{1}{\delta^2})$, which also holds for MRC and MMSE receivers. Therefore, substituting (\ref{eq:cdf_ZF}) into (\ref{eq:cdf_rate}), we have
\begin{align}
R_\mathrm{ZF} &= \frac{T_dN_t}{\ln(2)T}\int_0^{\frac{1}{\delta^2}} \frac{e^{-\frac{c_0\gamma}{1-\delta^2\gamma}}}{1+\gamma}\sum_{k=0}^{N_r-N_t}\frac{\left(\frac{c_0\gamma}{1-\delta^2\gamma}\right)^k}{k!}d\gamma\\
&\overset{(a)}=\frac{T_dN_t}{\ln(2)T}\sum_{k=0}^{N_r-N_t}\left(\frac{c_0}{\delta^2}\right)^k\int_0^{\infty}\frac{e^{-x}x^k}{(1+x)(\delta^2+(\delta^2+1)x)}dx\label{eq:anal_ZF_step2}\\
&\overset{(b)}= \frac{T_dN_t}{\ln(2)T}\sum_{k=0}^{N_r-N_t}\left(\frac{c_0}{\delta^2}\right)^k\left(\int_0^{\infty}\frac{e^{-x}x^k}{x+\frac{\delta^2}{\delta^2+1}}dx - \int_0^{\infty}\frac{e^{-x}x^k}{x+1}dx \right),\label{eq:anal_ZF_step3}
\end{align}
where $(a)$ is obtained by a change of variables $\frac{c_0\gamma}{1-\delta^2\gamma}\rightarrow x$ and exchanging the order of integration and summation, and $(b)$ is obtained by expanding the integrand in (\ref{eq:anal_ZF_step2}). Finally, applying \cite[Eq.~(3.383.10)]{gradshteyn2007a} in (\ref{eq:anal_ZF_step3}) and the fact that $E_n(z)=z^{n-1}\Gamma(1-n,z)$, where $\Gamma(n,z)$ is the upper incomplete gamma function \cite[Eq. (8.350.2)]{gradshteyn2007a}, concludes the proof. The achievable rate of MMSE receivers can be deduced in a similar way.
\subsection{Achievable rate of MRC receivers}
Substituting the SINR with MRC in (\ref{eq:cdf_ZF}) into (\ref{eq:cdf_rate}), and applying the change of  variables $\frac{c_0\gamma}{1-\delta^2\gamma}\rightarrow x$, as well as changing the order of summation and integration, we have
\begin{equation}\label{eq:anal_MRC_step1}
R_\mathrm{MRC} =  \frac{T_dN_t}{\ln(2)T}\sum_{k=0}^{N_r-1}\sum_{p=0}^{k}\alpha_{p,k}\left(\frac{c_0}{\delta^2+1}\right)^{p+N_t}\frac{1}{\delta^2}\int_0^{\infty}\frac{e^{-x}x^k}{\left(x + \frac{c_0}{\delta^2}\right)\left(x+\frac{c_0}{\delta^2+1}\right)^{p+N_t}}dx.
\end{equation}
Note that for any $a, b\in\mathbb C, a\neq b$, we can easily prove the following identity:
\begin{equation}
\frac{1}{(x+a)(x+b)^n} = \frac{(b-a)^{-n}}{x+a} - \sum_{i=1}^{n}\frac{(b-a)^{-i}}{(x+b)^{n-i+1}}.
\end{equation}
Applying this result to (\ref{eq:anal_MRC_step1}), and using the integral identity in \cite[Eq. (2.3.6.9)]{Prudnikov1986integrals} followed by some algebraic manipulations, we arrive at the desired result.

\section{Proof of Theorem \ref{theorem:bothlarge}}\label{appendix:bothlarge}
\subsection{Deterministic equivalent of ZF SINR}
The singular value decomposition (SVD) of $\bar{\mathbf H}$ is denoted by $\bar{\mathbf H} = \mathbf U\boldsymbol\Sigma\mathbf V^H$, where $\mathbf U$ and $\mathbf V$ are unitary matrices, while $\boldsymbol \Sigma$ is a $N_r\times N_t$ diagonal matrix with the diagonal elements containing the singular values of $\bar{\mathbf H}$. Consequently, we have
\begin{align}
\left[\left(\bar {\mathbf H}^H\bar {\mathbf H}\right)^{-1}\right]_{k,k} \!= \left[\mathbf V\boldsymbol\Lambda\mathbf V^H\right]_{k,k} = \mathbf v_k^H\boldsymbol\Lambda\mathbf v_k \underset{(a)}{\stackrel{a.s.}\longrightarrow} \frac{1}{N_t}\textrm{tr}\left(\boldsymbol\Lambda\right) \stackrel{a.s.}\longrightarrow \frac{1}{N_r\!-\!N_t},
\end{align}
where $\boldsymbol\Lambda$ is a diagonal matrix with its diagonal elements being the eigenvalues of $\left(\bar {\mathbf H}^H\bar {\mathbf H}\right)^{-1}$. Note that $\mathbf v_k$ is known to satisfy the conditions of $\mathbf x$ in Lemma \ref{lemma:TraceLemma} \cite{Eldar2003ZFAsym_TIT}, therefore $(a)$ follows readily. Substituting the above result into the ZF SINR expressions in (\ref{eq:SINR_ZF}) completes the proof.

\subsection{Deterministic equivalent of MRC SINR}
Rewriting the SINR of MRC, we have
\begin{align}
\gamma_{k,\mathrm{MRC}} &= \frac{\mid\bar{\mathbf h}_k^H\bar{\mathbf h}_k\mid^2}{\bar{\mathbf h}_k^H\left( (1+\delta^2)\bar{\mathbf H}\bar{\mathbf H}^H - \bar{\mathbf h}_k\bar{\mathbf h}_k^H + c_0\mathbf I_{N_r} \right)\bar{\mathbf h}_k}\\
&=\frac{\mid\bar{\mathbf h}_k^H\bar{\mathbf h}_k\mid^2}{(1+\delta^2)\bar{\mathbf h}_k^H\left(  \sum_{i=1, i\neq k}^{N_t}\bar{\mathbf h}_i\bar{\mathbf h}_i^H+ \frac{c_0}{1+\delta^2}\mathbf I_{N_r} \right)\bar{\mathbf h}_k + \delta^2\mid\bar{\mathbf h}_k^H\bar{\mathbf h}_k\mid^2}\\
&\underset{(a)}{\stackrel{a.s.}\longrightarrow} \frac{1}{\frac{1+\delta^2}{N_r^2}\mathrm{tr}\left( \sum_{i=1, i\neq k}^{N_t}\bar{\mathbf h}_i\bar{\mathbf h}_i^H+ \frac{c_0}{1+\delta^2}\mathbf I_{N_r} \right) + \delta^2}\\
&\underset{(b)}{\stackrel{a.s.}\longrightarrow}\frac{\beta}{1+\delta^2 + c_1 + \delta^2\beta},
\end{align}
where $(a)$ and $(b)$ are obtained by applying  Lemma \ref{lemma:TraceLemma}.
\subsection{Deterministic equivalent of MMSE SINR}
First, rewrite the MMSE SINR expression in (\ref{eq:MMSE_SINR_trasform1}) as
\begin{align}
\gamma_{k,\mathrm{MMSE}} = \frac{\tilde\gamma_{k,\mathrm{MMSE}} }{1+\delta^2+\delta^2\tilde\gamma_{k,\mathrm{MMSE}} },\label{eq:MMSE_SINR2}
\end{align}
where $\tilde\gamma_{k,\mathrm{MMSE}}\triangleq {\bar {\mathbf h}_k}^H\left(\sum_{i=1,i\neq k}^{N_t}{\bar {\mathbf h}_i}{\bar {\mathbf h}_i}^H + \frac{c_0}{1+\delta^2}  \mathbf I_{N_r}   \right)^{-1}{\bar {\mathbf h}_k}$.
As $N_r, N_t \rightarrow \infty$ with $\beta = \frac{N_r}{N_t} \in [1, \infty)$, we have
\begin{align}
\tilde\gamma_{k,\mathrm{MMSE}} {\stackrel{a.s.}\longrightarrow} m_{F_{\bar{\mathbf H}^H\bar{\mathbf H}}}\left(\frac{c_0}{1+\delta^2}\right) &= \mathrm{tr}\left(\left(\bar{\mathbf H}\bar{\mathbf H}^H + \frac{c_0}{1+\delta^2}\mathbf I_{N_r}\right)^{-1}\right)\label{eq:STtrans1}\\
&\overset{(a)}{=} \frac{1+\delta^2}{c_0}(N_r-N_t) + \mathrm{tr}\left(\left(\bar{\mathbf H}^H\bar{\mathbf H} + \frac{c_0}{1+\delta^2}\mathbf I_{N_t}\right)^{-1}\right)\notag\\
&\overset{(b)}{=} \frac{1+\delta^2}{c_0}(N_r-N_t) + \sum_{i=1}^{N_t}\frac{1+\delta^2}{c_0(1 + \bar{\mathbf h}_i^H\left(\bar{\mathbf H}\bar{\mathbf H}^H + \frac{c_0}{1+\delta^2}\mathbf I_{N_r}\right)^{-1}\bar{\mathbf h}_i)}\notag\\
&\underset{(c)}{\stackrel{a.s.}\longrightarrow}\frac{1+\delta^2}{c_0}(N_r-N_t) + \frac{(1+\delta^2)N_t}{c_0(1+m_{F_{\bar{\mathbf H}^H\bar{\mathbf H}}}(\frac{c_0}{1+\delta^2}))},\label{eq:STtrans4}
\end{align}
where $F_{\bar{\mathbf H}^H\bar{\mathbf H}}$ is the eigenvalue CDF of $\bar{\mathbf H}\bar{\mathbf H}^H$, and $m_{F_{\bar{\mathbf H}^H\bar{\mathbf H}}}\left(\frac{c_0}{1+\delta^2}\right)$ is the Stieltjes transform of $F_{\bar{\mathbf H}^H\bar{\mathbf H}}$ with argument $\frac{c_0}{1+\delta^2}$ \cite{Couillet2011a}, $(a)$ is obtained by using Lemma \ref{lemma:STLemma}, $(b)$ is obtained by using the matrix inversion identity (\ref{eq:matrixinversion}) and Lemma \ref{lemma:Rank1Lemma}, while $(c)$ is obtained by using Lemma \ref{lemma:TraceLemma}.

Combining (\ref{eq:STtrans1}) and (\ref{eq:STtrans4}), we can easily obtain
\begin{equation}
\tilde\gamma_{k,\mathrm{MMSE}} {\stackrel{a.s.}\longrightarrow} m_{F_{\bar{\mathbf H}^H\bar{\mathbf H}}}\left(\frac{c_0}{1+\delta^2}\right) = \frac{1}{2c_1}\left(-d + \sqrt{d^2 + \frac{4\beta c_1}{1+\delta^2}}  \right).
\end{equation}
Substituting this result into (\ref{eq:MMSE_SINR2}) and applying the continuous mapping theorem \cite[Theorem 2.3]{Van2000Asymptotic} concludes the proof.

%

\bibliographystyle{IEEEtran}
\bibliography{IEEEabrv,Reference}

\begin{thebibliography}{10}
\providecommand{\url}[1]{#1}
\csname url@samestyle\endcsname
\providecommand{\newblock}{\relax}
\providecommand{\bibinfo}[2]{#2}
\providecommand{\BIBentrySTDinterwordspacing}{\spaceskip=0pt\relax}
\providecommand{\BIBentryALTinterwordstretchfactor}{4}
\providecommand{\BIBentryALTinterwordspacing}{\spaceskip=\fontdimen2\font plus
\BIBentryALTinterwordstretchfactor\fontdimen3\font minus
  \fontdimen4\font\relax}
\providecommand{\BIBforeignlanguage}[2]{{%
\expandafter\ifx\csname l@#1\endcsname\relax
\typeout{** WARNING: IEEEtran.bst: No hyphenation pattern has been}%
\typeout{** loaded for the language `#1'. Using the pattern for}%
\typeout{** the default language instead.}%
\else
\language=\csname l@#1\endcsname
\fi
#2}}
\providecommand{\BIBdecl}{\relax}
\BIBdecl

\bibitem{xinlin2014training_ICC}
X.~Zhang, M.~Matthaiou, M.~Coldrey, and E.~Bj{\"o}rnson, ``Impact of residual
  transmit {RF} impairments on training-based {MIMO} systems,'' in
  \emph{Proc.~IEEE ICC}, June 2014, pp. 4752--4757.

\bibitem{telatar1999capacity}
E.~Telatar, ``Capacity of multi-antenna {Gaussian} channels,'' \emph{Europ.
  Trans. Telecom.}, vol.~10, no.~6, pp. 585--595, Nov.-Dec. 1999.

\bibitem{foschini1998a}
G.~J. Foschini and M.~J. Gans, ``On limits of wireless communications in a
  fading environment when using multiple antennas,'' \emph{Wireless Pers.
  Commun.}, vol.~6, no.~3, pp. 311--335, Mar. 1998.

\bibitem{hassibi2003}
B.~Hassibi and B.~M. Hochwald, ``How much training is needed in
  multiple-antenna wireless links?'' \emph{{IEEE} Trans. Inf. Theory}, vol.~49,
  no.~4, pp. 951--963, Apr. 2003.

\bibitem{biguesh2006training}
M.~Biguesh and A.~B. Gershman, ``Training-based {MIMO} channel estimation: {A}
  study of estimator tradeoffs and optimal training signals,'' \emph{{IEEE}
  Trans. Signal Process.}, vol.~54, no.~3, pp. 884--893, Mar. 2006.

\bibitem{mikael1}
M.~Coldrey and P.~Bohlin, ``Training-based {MIMO} systems--{P}art {I}:
  {P}erformance comparison,'' \emph{{IEEE} Trans. Signal Process.}, vol.~55,
  no.~11, pp. 5464--5476, Nov. 2007.

\bibitem{emil2010est}
E.~Bj{\"{o}}rnson and B.~Ottersten, ``A framework for training-based estimation
  in arbitrarily correlated {R}ician {MIMO} channels with {R}ician
  disturbance,'' \emph{{IEEE} Trans. Signal Process.}, vol.~58, no.~3, pp.
  1807--1820, Nov. 2010.

\bibitem{Love2014Training1_JSTSP}
J.~Choi, D.~J. Love, and P.~Bidigare, ``Downlink training techniques for {FDD}
  massive {MIMO} systems: Open-loop and closed-loop training with memory,''
  \emph{{IEEE} J. Sel. Topics Signal Process.}, vol.~8, no.~5, pp. 802--814,
  Oct. 2014.

\bibitem{Love2014Training2_JSTSP}
S.~Noh, M.~D. Zoltowski, Y.~Sung, and D.~J. Love, ``Pilot beam pattern design
  for channel estimation in massive {MIMO} systems,'' \emph{{IEEE} J. Sel.
  Topics Signal Process.}, vol.~8, no.~5, pp. 787--801, Oct. 2014.

\bibitem{schenk2008rf}
T.~Schenk, \emph{RF Imperfections in High-Rate Wireless Systems: Impact and
  Digital Compensation}.\hskip 1em plus 0.5em minus 0.4em\relax Springer, 2008.

\bibitem{studer2010residual}
C.~Studer, M.~Wenk, and A.~Burg, ``{MIMO} transmission with residual
  transmit-{RF} impairments,'' in \emph{Proc.~ITG/IEEE Work. Smart Ant. (WSA)},
  Feb. 2010, pp. 189--196.

\bibitem{Jingya2014IQI_TCOM}
J.~Li, M.~Matthaiou, and T.~Svensson, ``{I/Q} imbalance in {AF} dual-hop
  relaying: {P}erformance analysis in {N}akagami-m fading,'' \emph{{IEEE}
  Trans. Commun.}, vol.~62, no.~3, pp. 836--847, Mar. 2014.

\bibitem{Giuseppe2014PA_TCOM}
G.~Durisi, A.~Tarable, C.~Camarda, R.~Devassy, and G.~Montorsi, ``Capacity
  bounds for {MIMO} microwave backhaul links affected by phase noise,''
  \emph{{IEEE} Trans. Commun.}, vol.~62, no.~3, pp. 920--929, Mar. 2014.

\bibitem{Bussgang}
J.~J. Bussgang, ``Crosscorrelation functions of amplitude-distorted {Gaussian}
  signals,'' \emph{Technical Report No. 216, Research Laboratory of
  Electronics, Massachusetts Institute of Technology, Cambridge, MA}, Mar.
  1952.

\bibitem{Ochiai2002Clipped_TCOM}
H.~Ochiai and H.~Imai, ``Performance analysis of deliberately clipped {OFDM}
  signals,'' \emph{{IEEE} Trans. Commun.}, vol.~50, no.~1, pp. 89--101, Jan.
  2002.

\bibitem{Wenyi2012Framework_TCOM}
W.~Zhang, ``A general framework for transmission with transceiver distortion
  and some applications,'' \emph{{IEEE} Trans. Commun.}, vol.~60, no.~2, pp.
  384--399, Feb. 2012.

\bibitem{Ulf2014Impairments_GLOBECOM}
U.~Gustavsson, C.~Sanchez-Perez, T.~Eriksson, F.~Athley, G.~Durisi, P.~Landin,
  K.~Hausmair, C.~Fager, and S.~Lars, ``On the impact of hardware impairments
  on massive {MIMO},'' in \emph{Proc.~IEEE GLOBECOM}, Dec. 2014.

\bibitem{Emil2013imp_COML}
E.~Bj{\"{o}}rnson, P.~Zetterberg, M.~Bengtsson, and B.~Ottersten, ``Capacity
  limits and multiplexing gains of {MIMO} channels with transceiver
  impairments,'' \emph{{IEEE} Commun. Lett.}, vol.~17, no.~1, pp. 91--94, Jan.
  2013.

\bibitem{xinlin2014capacity_ICC}
X.~Zhang, M.~Matthaiou, E.~Bj{\"o}rnson, M.~Coldrey, and M.~Debbah, ``On the
  {MIMO} capacity with residual transceiver hardware impairments,'' in
  \emph{Proc.~IEEE ICC}, June 2014, pp. 5310--5316.

\bibitem{Michalis2013Dhop_TCOM}
E.~Bj{\"o}rnson, M.~Matthaiou, and M.~Debbah, ``A new look at dual-hop
  relaying: Performance limits with hardware impairments,'' \emph{{IEEE} Trans.
  Commun.}, vol.~61, no.~11, pp. 4512--4525, Nov. 2013.

\bibitem{Michalis2013TwoWay_COML}
M.~Matthaiou, A.~Papadogiannis, E.~Bj{\"o}rnson, and M.~Debbah, ``Two-way
  relaying under the presence of relay transceiver hardware impairments,''
  \emph{{IEEE} Commun. Lett.}, vol.~17, no.~6, pp. 1136--1139, June 2013.

\bibitem{emil2013est}
E.~Bj{\"{o}}rnson, J.~Hoydis, M.~Kountouris, and M.~Debbah, ``Massive {MIMO}
  systems with non-ideal hardware: Energy efficiency, estimation, and capacity
  limits,'' \emph{{IEEE} Trans. Inf. Theory}, no.~11, pp. 7112--7139, Nov.
  2014.

\bibitem{Jorswieck2004Transmission_TSP}
E.~A. Jorswieck and H.~Boche, ``Optimal transmission strategies and impact of
  correlation in multiantenna systems with different types of channel state
  information,'' vol.~52, no.~12, pp. 3440--3453, Dec. 2004.

\bibitem{holma2011LTE}
H.~Holma and A.~Toskala, \emph{{LTE} for {UMTS}: {E}volution to
  {LTE}-{A}dvanced}.\hskip 1em plus 0.5em minus 0.4em\relax Wiley, 2011.

\bibitem{Agilent8Hints}
``Agilent 8 hints for making and interpreting {EVM} measurements: application
  note,'' {A}vailable [Online]:
  http://cp.literature.agilent.com/litweb/pdf/5989-3144EN.pdf.

\bibitem{kaybook1}
S.~M. Kay, \emph{Fundamentals of Statistical Signal Processing, Volume 1:
  Estimation theory}.\hskip 1em plus 0.5em minus 0.4em\relax Prentice Hall PTR,
  1993.

\bibitem{Mckay2009LinearReceiver}
R.~H. Louie, M.~R. McKay, and I.~B. Collings, ``Maximum sum-rate of {MIMO}
  multiuser scheduling with linear receivers,'' \emph{{IEEE} Trans. Commun.},
  vol.~57, no.~11, pp. 3500--3510, Nov. 2009.

\bibitem{gradshteyn2007a}
I.~S. Gradshteyn and I.~M. Ryzhik, \emph{Table of Integrals, Series, and
  Products}, 7th~ed.\hskip 1em plus 0.5em minus 0.4em\relax Academic {P}ress,
  2007.

\bibitem{winters1994diversity}
J.~H. Winters, J.~Salz, and R.~D. Gitlin, ``The impact of antenna diversity on
  the capacity of wireless communication systems,'' \emph{{IEEE} Trans.
  Commun.}, vol.~42, no. 234, pp. 1740--1751, Aug. 1994.

\bibitem{Ping2006MMSE}
P.~Li, D.~Paul, R.~Narasimhan, and J.~Cioffi, ``On the distribution of {SINR}
  for the {MMSE} {MIMO} receiver and performance analysis,'' \emph{{IEEE}
  Trans. Inf. Theory}, vol.~52, no.~1, pp. 271--286, Jan. 2006.

\bibitem{Marzetta2010Massive_TWCOM}
T.~L. Marzetta, ``Noncooperative cellular wireless with unlimited numbers of
  base station antennas,'' \emph{{IEEE} Trans. Commun.}, vol.~9, no.~11, pp.
  3590--3600, Nov. 2010.

\bibitem{Hoydis2013Massive}
J.~Hoydis, S.~{t}en Brink, and M.~Debbah, ``Massive {MIMO} in the {UL}/{DL} of
  cellular networks: {H}ow many antennas do we need?'' \emph{{IEEE} J. Sel.
  Areas Commun.}, vol.~31, no.~2, pp. 160--171, Feb. 2013.

\bibitem{Bai1993Smallest}
Z.~D. Bai and Y.~Q. Yin, ``Limit of the smallest eigenvalue of a large
  dimensional sample covariance matrix,'' \emph{Ann. Probab.}, vol.~21, no.~3,
  pp. 1275--1294, July 1993.

\bibitem{Eldar2003ZFAsym_TIT}
Y.~C. Eldar and A.~M. Chan, ``On the asymptotic performance of the
  decorrelator,'' \emph{{IEEE} Trans. Inf. Theory}, vol.~49, no.~9, pp.
  2309--2313, Sept. 2003.

\bibitem{Van2000Asymptotic}
A.~W. Van~der Vaart, \emph{Asymptotic statistics}.\hskip 1em plus 0.5em minus
  0.4em\relax Cambridge {U}niversity {P}ress, 2000.

\bibitem{Billingsley2008Probability}
P.~Billingsley, \emph{Probability and measure}.\hskip 1em plus 0.5em minus
  0.4em\relax John Wiley \& Sons, 2008.

\bibitem{Boyd2004Convex}
S.~P. Boyd and L.~Vandenberghe, \emph{Convex optimization}.\hskip 1em plus
  0.5em minus 0.4em\relax Cambridge {U}niversity {P}ress, 2004.

\bibitem{Jakob2011Training_TSP}
J.~Hoydis, M.~Kobayashi, and M.~Debbah, ``Optimal channel training in uplink
  network {MIMO} systems,'' \emph{{IEEE} Trans. Signal Process.}, vol.~59,
  no.~6, pp. 2824--2833, June 2011.

\bibitem{Bai2009RMBook}
Z.~Bai and J.~W. Silverstein, \emph{Spectral {A}nalysis of {L}arge
  {D}imensional {R}andom {M}atrices}.\hskip 1em plus 0.5em minus 0.4em\relax
  Springer, 2009.

\bibitem{Bai1995Empirical}
J.~W. Silverstein and Z.~Bai, ``On the empirical distribution of eigenvalues of
  a class of large dimensional random matrices,'' \emph{Journal of Multivariate
  analysis}, vol.~54, no.~2, pp. 175--192, Aug. 1995.

\bibitem{Couillet2011a}
R.~Couillet and M.~Debbah, \emph{Random Matrix Methods for Wireless
  Communications}.\hskip 1em plus 0.5em minus 0.4em\relax Cambridge
  {U}niversity {P}ress, 2011.

\bibitem{Shah1998MRC}
A.~Shar and A.~M. Haimovich, ``Performance analysis of optimum combining in
  wireless communications with {R}ayleigh fading and cochannel interference,''
  \emph{{IEEE} Trans. Commun.}, vol.~46, no.~4, pp. 473--479, Apr. 1998.

\bibitem{Gao1998MMSE}
H.~Gao, P.~J. Smith, and M.~V. Clark, ``Theoretical reliability of {MMSE}
  linear diversity combining in {R}ayleigh-fading additive interference
  channels,'' \emph{{IEEE} Trans. Commun.}, vol.~46, no.~5, pp. 666--672, May
  1998.

\bibitem{Prudnikov1986integrals}
A.~Prudnikov, Y.~A. Brychkov, and O.~Marichev, ``Integrals and series, vol. 1:
  Elementary functions,'' \emph{Gordon and Breach Science Publishers}, 1986.

\end{thebibliography}

\end{document}